\documentclass[11pt,a4paper]{article}
\usepackage{amsfonts, amsmath, amsthm, amssymb}
\usepackage{enumerate}
\usepackage{epsfig} 
\usepackage{graphicx}
\usepackage[T1]{fontenc} 
\usepackage[english]{babel}
\usepackage[applemac]{inputenc}
\usepackage{cite}
\usepackage{dsfont}
\usepackage{xcolor}

\usepackage[arrow, matrix, curve]{xy}

\newcommand{\M}{\mathcal P}
\newcommand{\IR}{\mathbb R}
\newcommand{\IP}{\mathbb P}
\newcommand{\IE}{\mathbb E}
\newcommand{\IN}{\mathbb N}

\newcommand{\dd}{\mathrm{d}}

\newcommand{\weak}{\rightharpoonup}

\theoremstyle{plain}
\newtheorem{Theorem}{Theorem}[section]
\newtheorem{Lemma}[Theorem]{Lemma}
\newtheorem{Corollary}[Theorem]{Corollary}
\newtheorem{Proposition}[Theorem]{Proposition}
\theoremstyle{definition}
\newtheorem{Definition}[Theorem]{Definition}
\newtheorem{Remarks}[Theorem]{Remarks}

\begin{document}
\title{A mean field limit for the Vlasov-Poisson system}
\author{Dustin Lazarovici\thanks{dustin.lazarovici@live.com}\; and Peter Pickl\thanks{pickl@math.lmu.de}\\[1.5ex]
Mathematisches Institut, Ludwig-Maximilians Universit\"at\\ 
Theresienstr. 39, 80333 Munich, Germany.}

\maketitle

\abstract{\noindent We present a probabilistic proof of the mean field limit and propagation of chaos $N$-particle systems in three dimensions with positive (Coulomb) or negative (Newton) $1/r$ potentials scaling like $1/N$ and an $N$-dependent cut-off which scales like $N^{-1/3+ \epsilon}$. In particular, for typical initial data, we show convergence of the empirical distributions to solutions of the Vlasov-Poisson system with either repulsive electrical or attractive gravitational interactions.}

\section{Introduction}
\noindent We are interested in a microscopic derivation of the nonrelativistic Vlasov-Poisson system. This equation describes a plasma of identical charged particles with electrostatic or gravitational interactions 
\begin{align}\label{VP}
	&\partial_t f + p \cdot \nabla_q f + (k* \rho_t) \cdot \nabla_p f = 0, 
\end{align}

\noindent where $k$ is the (Coulomb) kernel

\begin{equation}\label{Coulombkernel3D} k(q) := \sigma\frac{q}{\,\lvert q \rvert ^{3}},\hspace{4mm} \sigma=\lbrace \pm 1 \rbrace,\end{equation}
and
\begin{equation}  \rho_t(q) =  \rho[f_t](q) = \int  f(t,q,p) \,  \mathrm{d}^3p \end{equation}
\noindent is the charge density induced by the distribution $f(t,p,q) \geq 0$. 

\noindent Units are chosen such that all constants, in particular the mass and charge of the particles, are equal to $1$.  The case $\sigma = +1$ corresponds to electrostatic (repulsive) interactions, while $\sigma = -1$ describes gravitational (attractive) interactions. In the gravitational case, (1-3) is also known as the \emph{Vlasov-Newton} equation. For simplicity, we focus on the $3$-dimensional case, generalization to arbitrary dimensions $d \geq 2$ is straightforward.

\subsection{Previous results} 
While the existence theory of the Vlasov-Poisson dynamics is well understood -- we will cite the pertinent results below -- its microscopic derivation is still an open problem. To our knowledge, the first paper to discuss a mathematically rigorous derivation of Vlasov equations is Neunzert and Wick, 1974 \cite{NeunzertWick}. Better known are the publications of Braun and Hepp, 1977 \cite{BraunHepp} and Dobrushin, 1979 \cite{Dobrushin}, as well as the later exposition of Neunzert, 1984 \cite{Neunzert}. For a general overview of the topic, we refer the reader to the book of Spohn \cite{SpohnBook}.

Rather than the Vlasov-Poisson equation, the papers of Neunzert, Braun and Hepp and Dobrushin treat simplified models with Lipschitz continuous forces $k \in W^{1,\infty} = \lbrace k \in C^1(\IR^d) : \lVert k \rVert_\infty + \lVert \nabla k \rVert_\infty < \infty \rbrace $. The last few years have seen great progress in treating mean field limits for singular forces up to but not including the Coulomb case. In particular, Hauray and Jabin, 2013,  discuss force kernels bounded as $\lvert k(q) \rvert \leq \frac{C}{\lvert q \rvert^\alpha}$ with $\alpha <  d -1$ in $d \geq 3$ dimensions \cite{HaurayJabin}. For $1 < \alpha < d-1$, they perform the mean field limit for typical initial data and an $N$-dependent cut-off that can be chosen as small as $N^{-\frac {1}{2d}}$ for $\alpha \nearrow d-1$. For $\alpha < 1$, they are even able to prove molecular chaos with no cut-off at all. Unfortunately, their method fails precisely at the Coulomb threshold $\alpha = d-1$. 
	
In contrast, Kiessling, 2014 proves a non-quantitative approximation result including the Coulomb singularity under the assumption of an (uniform in $N$) a priori bound on the microscopic forces. The status of this assumption, however, whether it is satisfied for generic initial data or not, remains open \cite{Kiessling}.

Recently, Boers and Pickl proposed a novel method for deriving mean field equations which is designed for stochastic initial conditions, thus aiming directly at a typicality result. With this method, they were able to improve the cut-off near the Coulomb case to $\sim N^{-\frac{1}{d}}$ \cite{Peter}. 
	
The aim of this paper is to extend the method of Boers and Pickl to include the Coulomb singularity in the large $N$ limit, thus aiming at a microscopic derivation of the Vlasov-Poisson dynamics. The Coulomb case is qualitatively different from the previously treated interactions since the mean field force $k * \rho$ is no longer Lipschitz, in general, even for bounded $\rho$. However, we will show how this critical case can be treated by exploiting the second order nature of the dynamics and introducing an anisotropic scaling of the relevant metric. Moreover, we optimize the method in such a way as to achieve a rate of convergence that can be faster than any inverse power of $N$, depending on decay properties of the initial distribution $f_0$.

An alternative proof, based on similar modifications of the Wasserstein distance, is simultaneously proposed in \cite{Dustin}. The result presented here, however, allows for a significantly smaller cutoff $N^{-\frac{1}{d}+\epsilon}$ to be compared with $N^{-\frac{1}{d(d+2)}+\epsilon}$.

\section{The microscopic model}
Since the Coulomb kernel is strongly singular at the origin, we will require a regularization on the microscopic level. For $N \in \IN$ and $\delta \geq 0$, we consider
\begin{equation}\label{regularizedkernel} k^N_\delta(q) := \sigma \begin{cases}\;\, \frac{q}{\,\lvert q \rvert^3}  & \text{, if}\; \lvert q \rvert \geq N^{-\delta} \\[1.5ex]  q N^{3\delta} & \text{, else.}
\end{cases}\end{equation}
On $\IR^3 \setminus \lbrace 0 \rbrace$ this converges to the Coulomb kernel  \eqref{Coulombkernel3D} as $N \to \infty$. Of course, the $N$-dependence of the force thus introduced is a technical necessity rather than a realistic physical model, though similar regularizations are commonly used in numerical computations.

\noindent  In the \textit{mean field scaling}, the equations of motion for the regularized $N$-particle system are given by 

\begin{equation}\label{microscopiceq1} \begin{cases} \dot{q_i}(t) = p_i(t)\\[1.2ex] \dot{p_i}(t) = \frac{1}{N} \sum\limits_{j=1}^{N} k_\delta^N(q_i - q_j),  \end{cases}\end{equation}

\noindent for $i \in {1,..., N}$. Since the vector field is Lipschitz for fixed $\delta, N$, we have global existence and uniqueness of solutions and hence an $N$-particle Hamiltonian flow which we denote by $^{N}\Psi_{t,s}(Z) = \bigl( {}^N\Psi^{1}_{t,s}(Z), {}^N\Psi^{2}_{t,s}(Z) \bigr) \in \IR^{3N}\times\IR^{3N}$. Introducing the $N$-particle force $K^N_\delta:\IR^{3N}\to \IR^{3N}$ given by
\begin{equation}\label{microscopicflow} (K^N_\delta(q_1,..,q_N))_i:= \frac{1}{N} \sum\limits_{j=1}^{N} k_\delta^N(q_i - q_j), \; \; i = 1,..,N, \end{equation} 
we can also characterize ${}^N\Psi_{t,s}$ as the solution of 
\begin{align}\label{microNeq} \frac{\dd}{\dd t}  \bigl(\Psi^{1}_{t,s}(Z), \Psi^{2}_{t,s}(Z) \bigr) =  \bigl(\Psi^{2}_{t,s}(Z), K^N_\delta(\Psi^{1}_{t,s}(Z)) \bigr), \;\Psi_{s,s}(Z)=Z.\end{align}

\noindent Finally, if $^{N}\Psi_{t,0}(Z) = (q_i(t),p_i(t))_{i=1,..,N}$, we define the corresponding \emph{microscopic} or \emph{empirical density} by
\begin{equation} \mu^N_t[Z] = \mu^N_0[\Psi_{t,0}(Z)] := \frac{1}{N} \sum \limits_{i=1}^N \delta(\cdot - q_i(t)) \delta(\cdot - p_i(t)).
\end{equation}
Our aim is to show that for \emph{typical} $Z$, the empirical density $\mu^N_t[Z]$ converges to a solution $f_t$ of the Vlasov-Poisson equation as $N \to \infty$.\\

\noindent Of course, more general cut-offs can be considered. In the literature, the following nomenclature has been established (see e.g. \cite{HaurayJabin}): 
\begin{Definition}\label{Def:conditions}
	A pair-interaction defined by a kernel $k: \IR^d \to \IR^d$ satisfies a \emph{$S^\alpha$-condition}, if
	\begin{itemize}
		\item[] 	$(S^\alpha) \hspace{.7cm}  \exists c>0, \forall q \in \IR^d\setminus\lbrace 0 \rbrace \;\; \lvert k(q) \rvert \leq \frac{c}{\lvert q \rvert^\alpha}, \;\; \lvert \nabla k \rvert \leq \frac{c}{\lvert q\rvert^{\alpha +1}}.$
	\end{itemize}
	Introducing a cut-off of order $N^{-\delta}$ near the origin, the regularized force kernel $k^N_\delta$ satisfies a \textit{$(S^\alpha_\delta)$-condition} if
	
	\begin{itemize}
		\item[] \hspace{1.5cm} $i) \;\;\;\; k \text{ satisfies a } (S^\alpha)$ condition,
		\item[] \hspace{0cm} $(S^\alpha_\delta) \hspace*{.7cm} ii)\;\; \; k^N_\delta(q) =  k(q) \text{ for } \lvert q \rvert \geq N^{-\delta}$,
		\item[] \hspace{1.4cm} $iii)\; \lvert k^N_\delta(q) \rvert \leq N^{\delta \alpha} \text{ for all } \lvert q \rvert < N^{-\delta}$.
	\end{itemize}
	In addition, we shall require that
	\begin{equation} \label{additionalassumption}
	\hspace{0.5cm} iv) \; \lvert \nabla k^N_\delta(q) \rvert \leq N^{\delta ({\alpha}+1)} \text{ for all } \lvert q \rvert < N^{-\delta},
	\end{equation}
	which assures that the regularization around the origin is not too erratic.
\end{Definition}
\noindent Within this setting, we thus consider 3-dimensional force kernels satisfying a $(S^\alpha_\delta)$ condition with $\alpha = 2$ and the additional assumption $iv)$. The lower bound on the cut-off will later be determined as $\delta < \frac{1}{3}$. Moreover, we shall adopt the convention $k^N_\delta(0)=0$, meaning that the microscopic dynamics do not contain self-interactions. The reader is free to think of \eqref{regularizedkernel} as defining the microscopic model or consider another regularization of his liking that satisfies the above assumptions.

\subsection{The regularized Vlasov-Poisson equation}
\noindent For any $\delta > 0$ and $N \in \IN\cup\lbrace \infty \rbrace$, we also consider the corresponding mean field equation
\begin{equation}\label{Vlasov}
	\partial_t f + p \cdot \nabla_q f + \Bigl(k^N_\delta * \rho_t \Bigr) \cdot \nabla_p f = 0. 
\end{equation}
\noindent For (formally) $N = \infty$,  this reduces to the Vlasov-Poisson equation \eqref{VP}. For a fixed initial distribution $f_0 \in L^\infty(\IR^3\times \IR^3)$ with $f_0 \geq 0$ and $\int f = 1$ we denote by $f^N_t$ the unique solution of  \eqref{Vlasov} with initial datum $f^N_t(0, \cdot,\cdot) = f_0$.

\subsection{Method of characteristics}

\noindent It is convenient to consider the \emph{characteristic flow} of the mean field system. For $N \in \IN, \delta >0$ and $\rho \in L^1(\IR^3)$, we define $\widehat{K}^N_\delta(\cdot; \rho): \IR^3\times \IR^3 \to \IR^3 \times \IR^3$ by
\begin{equation}
	\widehat{K}^N_\delta (q,p; \rho) := \bigl(p, k^N_\delta*\rho\,(q)\bigr).
\end{equation} 

\noindent Then, the (regularized) Vlasov-Poisson equation \eqref{Vlasov} with initial $f_0$ is equivalent to the following system of integro-differential equations: 

\begin{equation}\label{chareq2} \begin{cases} \frac{\dd}{\dd t} {\varphi^N_{t,s}}(z; f_0) = \widehat{K}^N_\delta\bigl(\varphi^N_{t,s}(z; f_0); \rho^N_t \bigr) \\[1.1ex]
		\rho^N_t(q) = \int  f^N(t, q, p) \, \mathrm{d}^3p\\[1.1ex]
		f^N(t, \cdot) = \varphi^N_{t,s}(\cdot\,; f_0) \# f^N_s \\[1.1ex]
		\varphi^N_{s,s}(z; f_0)=z . \end{cases}\end{equation}

 \noindent Here, $\varphi(\cdot) \# f$ denotes the image-measure of $f$ under $\varphi$, defined by  $\varphi \#f(A) = f(\varphi^{-1}(A)) $ for any Borel set $A \subseteq \IR^6$.\\

\noindent In other words, we have non-linear time-evolution in which $\varphi^N_{t,s}(\cdot\,; f_0)$ is the one-particle flow induced by the mean field dynamics with initial distribution $f_0$, while, in turn, $f_0$ is transported with the flow $\varphi^N_{t,s}$. Due to the semi-group property $\varphi^N_{t,s'}\circ\varphi^N_{s',s} = \varphi^N_{t,s}$ it generally suffices to consider the initial time $s=0$.\\

\noindent The method of characteristics can also be though of as establishing a kind of duality between the (rescaled) Newtonian dynamics \eqref{microscopiceq1} and the Vlasov equation \eqref{Vlasov}. Indeed, observing that the microscopic force can be written as
\begin{equation} \frac{1}{N}\sum\limits_{j =1}^N k^N_\delta(q_i - q_j) = k^N_\delta * \mu^N_t[Z] (q_i), \end{equation}
one easily checks that $\Psi_{t,0}(Z)$ solves \eqref{microscopiceq1} with $\Psi_0(Z) = 0$ if and only if $g_t = \mu^N_0[\Psi_{t,0}(Z)]$ is a weak solution of \eqref{Vlasov} with $g_0 = \mu^N_0[Z]$. 

This relation is often used to translate the microscopic dynamics into a Vlasov equation, allowing to treat $\mu^N_t[Z]$ and $f_t$ on the same footing.  Here, we will go the opposite way, so to speak, and transform the mean field dynamics into corresponding $N$ particle dynamics. \noindent To this end, we consider the lift of $\varphi^N_{t,s}(\cdot)$ to the $N$-particle phase-space, which we denote by ${}^N\Phi_{t,s}$. That is, for $f_0 \in L^1(\IR^6)$ and $Z= (q_i, p_i)_{1\leq i \leq N}$, we define

\begin{equation}\label{meanfieldflow} {}^N\Phi_{t,s}(Z;f_0) := \bigl(\varphi^N_{t,s}(q_1,p_1; f_0), ..., \varphi^N_{t,s}(q_N,p_N; f_0) \bigr). \end{equation}

\noindent Denoting by $\overline{K}:\IR^{3N}\to\IR^{3N}$ the lift of the mean field force to the $N$-particle phase-space, i.e. 
\begin{equation}\label{Kbar} (\overline{K}_t(Z))_i := k^N_\delta*\rho[f^N_t](z_i),\;\;\; Z=(z_1,...,z_N),\end{equation}
the flow ${}^N\Phi_{t,s}(Z) = \bigl({}^N\Phi^{1}_{t,s}(X), {}^N\Phi^{2}_{t,s}(X) \bigr)$ can also be characterized as the solution of the non-autonomous differential equation
\begin{align} \frac{\dd}{\dd t}  \bigl(\Phi^{1}_{t,s}(Z), \Phi^{2}_{t,s}(Z) \bigr) =  \bigl(\Phi^{2}_{t,s}(Z), \overline{K}_t(\Phi^{1}_{t,s}(Z)) \bigr), \; \Phi_{s,s}(Z)=Z \end{align}
to be compared with \eqref{microNeq}. Finally, we introduce the corresponding empirical density
\begin{equation} \mu^N_0[\Phi_{t,0}(Z)] = \varphi^N_{t,0} \# \mu^N_0[Z], \end{equation}
pertaining to the mean field dynamics with (random) initial conditions $Z \in \IR^{6N}$.\\

\noindent In summary, for fixed $f_0$ and $N \in \IN$, we consider for any initial configuration $Z \in \IR^{6N}$  two different time-evolutions: ${}^N\Psi_{t,0}(Z)$, given by the microscopic equations \eqref{microscopiceq1} and ${}^N\Phi_{t,0}(Z)$, given by the time-dependent mean field force generated by $f^N_t$. We are going to show that for typical $Z$, the two time-evolutions or close in an appropriate sense.

\section{Existence of solutions}\label{section:VPexistenceofsolutions}

\noindent For the well-posedness of the Vlasov-Poisson system, we can rely on various results establishing global existence and uniqueness of (weak and strong) solutions under fairly mild conditions on the initial configuration $f_0$ (Pfaffelmoser, 1990 \cite{Pfaffelmoser}, Schaeffer, 1991\cite{Schaeffer}, Lions and Perthame, 1991\cite{LionsPerthame}, Horst, 1993 \cite{Horst}). For our purposes, the following existence result due to Lions and Perthame is particularly useful:

\begin{Theorem}[Lions and Perthame]\label{Thm:LP}\mbox{}\\
	Let $f_0 \geq 0, f_0 \in  L^1(\IR^3\times \IR^3) \cap L^\infty(\IR^3 \times \IR^3)$ satisfy
	
	\begin{equation} \int \lvert p \rvert^m f_0(q, p)  \,\dd q\,  \dd p < +\infty,\end{equation}
	for all $m<m_0$ and some $m_0 >3$. 
	\begin{enumerate}[a)]
		\item Then, the Vlasov-Poisson system defined by equations (1--3) has a continuous, bounded solution $f(t, \cdot,\cdot) \in C(\IR^+;L^p(\IR^3\times\IR^3)) \cap L^\infty(\IR^+;L^\infty(\IR^3 \times \IR^3))$ for $1 \leq p < \infty$ satisfying
		\begin{equation} \sup\limits_{t \in [0,T]} \int \lvert p \rvert^m  f(t,q, p)  \,\dd p \, \dd p < +\infty, \end{equation}
		for all $T < \infty, m < m_0$.\\
		
		\item If, in fact, $m_0 > 6$ and we assume that $f_0$ satisfies 
		\begin{equation}\label{Assumption} \begin{split}
		\mathrm{esssup} \lbrace f_0(q'+p t, p' ) : \lvert q-q'\rvert \leq Rt^2, \lvert p-p' \rvert < Rt \rbrace\\ 
		\in L^\infty\bigl((0,T)\times \IR^3_q; L^1(\IR^3_p)\bigr) \end{split} \end{equation}
		for all $ R >0$ and $T>0$, then
		\begin{equation}\label{rhobound}\sup_{t \in [0,T]} \lVert \rho_t(q) \rVert_\infty < + \infty, \; \forall\, T \in (0, +\infty). \end{equation}
	\end{enumerate}
\end{Theorem}

\noindent Under the assumption of part $b)$ of the theorem, the uniqueness result of Loeper, 2006 \cite{Loeper} then shows that for any $T > 0$, said $f$ is the \emph{unique} solution in the set of bounded, positive measures on $[0,T) \times \IR^6$ satisfying $f \bigl\lvert_{t=0} = f_0$ in the sense of distributions. Moreover, it is known that as long as the charge density is bounded, solutions with smooth initial data remain smooth (see e.g. in \cite{Horst3}).\\

\noindent As Lions and Perthame remark  -- and as one can verify by following their proof -- part $b)$ of the theorem actually yields a bound on the charge density that is uniform in $N$ if one considers a sequence of regularized time-evolutions as (for instance) in \eqref{Vlasov}. We will note this important fact in the following lemma. 

\begin{Lemma}\label{Lemma:uniformbound}
	Let $f_0 \in  L^1(\IR^3\times \IR^3) \cap L^\infty(\IR^3 \times \IR^3)$ and $f^N_t$ be the solution of the regularized Vlasov-Poisson equation \eqref{Vlasov} (with corresponding cut-off) and initial datum $f^N(0, \cdot, \cdot) = f_0$. If $f_0$ satisfies assumption \eqref{Assumption} of the above theorem, there exists a constant $C_\rho>0$ such that 
	\begin{equation}\label{Crho}\lVert \rho_t^N \rVert_\infty + \lVert \rho_t^N \rVert_1 \leq C_\rho, \;  \forall N\in \IN\cup\lbrace \infty\rbrace, \; \forall t >0, \end{equation}
	where $\rho^N_t=\rho[f^N_t]$ and, with a little abuse of notaiton,  $\rho^\infty_t =\rho[f_t]$.
\end{Lemma}


\noindent Since condition \eqref{Assumption} is rather abstract, we want to state a more intuitive sufficient criterion.

\begin{Lemma}
	Let $f_0 \in  L^1(\IR^3\times \IR^3) \cap L^\infty(\IR^3 \times \IR^3), \; f \geq 0.$ Suppose there exist functions $\rho \in L^\infty(\IR^3)$ and $\vartheta(\lvert p \rvert) \in L^1(\IR^3)$  with $\vartheta$ monotonously decreasing  and an $S >0$ such that for all $\lvert p \rvert > S$
	\begin{equation*}
	f_0(q,p) \leq \rho(q)\vartheta(\lvert p \rvert).
	\end{equation*}
	Then $f_0$ satisfies assumption \eqref{Assumption}. Special cases: 
	\begin{itemize}
		\item $f_0$ has compact support in the $p$-variables.
		\item $f_0$ is a thermal state of the form $\rho(q)\, e^{-\beta p^2}$ with $\lVert \rho \lVert_\infty < \infty, \beta >0$.
	\end{itemize}
\end{Lemma}

\begin{proof}
	For given $R, t >0$ we have to consider the function
	\begin{equation*} \tilde{f}(t,q,p):= \mathrm{esssup} \lbrace f_0(q'+p t, p' ) : \lvert q-q'\rvert \leq Rt^2, \lvert p-p' \rvert < Rt \rbrace.\end{equation*}
	Choosing $R' > S +RT$, we have
	\begin{align*}  &\int\limits_{\IR^3}  \tilde{f}(t,q,p) \, \dd^3 p = \int\limits_{\lvert p \rvert \leq R'} + \int\limits_{\lvert p \rvert > R'}  \tilde{f}(t,q,p) \, \dd^3 p\\
	\leq & \frac{4}{3}\pi R'^3 \lVert \tilde{f}(t,\cdot,\cdot) \rVert_{\infty} + \lVert \rho \rVert_\infty \int  \sup\limits_{\lvert p-p' \rvert < Rt}\vartheta (\lvert p'\rvert) \, \dd^3 p\\
	\leq & \frac{4}{3}\pi R'^3 \lVert f_0 \rVert_{\infty} + \lVert \rho \rVert_\infty \int \vartheta (\lvert p\rvert - Rt) \, \dd^3 p\\ 
	\leq & C \lVert f_0 \rVert_\infty + \lVert \rho \rVert_\infty \lVert \vartheta \rVert_{1} < \infty,
	\end{align*}
	where in the second to last line we used the monotonicity of $\vartheta(\lvert p \rvert)$ and the fact that $\lVert \tilde f \rVert_\infty = \lVert f_0 \rVert_\infty$.
\end{proof}

\noindent One crucial consequence of the bounded density is that the mean field force remains bounded, as well.

\begin{Lemma}\label{Lemma:forcebound}
	Let $k$ be the Coulomb kernel, and $\rho \in L^1\cap L^\infty(\IR^3; \IR^+)$. Then there exists $C>0$ such that
	\begin{equation}\label{forcebound} \lVert k * \rho \rVert_\infty \leq C\lVert \rho\rVert_1^{1/3} \lVert \rho \rVert_\infty^{2/3}. \end{equation}
	
\end{Lemma}
\begin{proof} For $R>0$, we compute: 
	\begin{align*} \lVert k * \rho \rVert_\infty &\leq \Bigl \lVert \int\limits_{\lvert y \rvert < R} k(y) \rho(x-y) \, \mathrm{d}^3y \Bigr \rVert_\infty + \Bigl \lVert \int\limits_{\lvert y \rvert \geq R} k(y) \rho(x-y) \, \mathrm{d}^3y \Bigr \rVert_\infty\\
	&\leq \lVert \rho \rVert_\infty \,  \int\limits_{\lvert y \rvert <R} \frac{1}{\lvert y \rvert^{2}} \mathrm{d}^3 y + R^{-2}\lVert \rho \rVert_1 = 4 \pi  R \lVert \rho \rVert_\infty + R^{-2}\lVert \rho \rVert_1.
	\end{align*}
	This last expression is optimized by setting $R = (4\pi)^{-1/3}\lVert \rho \rVert_\infty ^{-1/3} \lVert \rho\rVert_1^{1/2}$, which yields 
	$\lVert k * \rho \rVert_\infty \leq 2( 4 \pi)^{2/3}  \lVert \rho\rVert_1^{1/3}\lVert \rho \rVert_\infty^{2/3}.$
	
\end{proof}

\section{Statement of the results}
In the following, all probabilities and expectation values are meant with respect to the product measure given at a certain time by $f^N_t$. That is, for any random variable $H:\mathbb{R}^{6N}\to\mathbb{R}$ and any element $A$ of the Borel algebra  
\begin{align}
\mathbb{P}^N_t(H\in A)=& \int_{H^{-1}(A)} \prod_{j=1}^N f^N_t(z_j)dZ\\
\mathbb{E}^N_t(H)=& \int_{\mathbb{R}^{6N}} H(Z)\prod_{j=1}^N f^N_t(z_j)dZ\;.
\end{align}

\noindent Note that since $^{N}\Phi_{t, s}$ leaves the measure invariant,
\begin{align*}
\mathbb{E}^N_s(H\circ {}^{N}\Phi_{t, s})
=&
\int_{\mathbb{R}^{6N}} H(^{N}\Phi_{t, s}(Z))\prod_{j=1}^N f^N_s(z_j)dZ
\\
=&\int_{\mathbb{R}^{6N}} H(Z)\prod_{j=1}^N f^N_s(\varphi^N_{s,t}(z_j))dZ\\
=&\int_{\mathbb{R}^{6N}} H(Z)\prod_{j=1}^N f^N_t(z_j)dZ= \mathbb{E}^N_t(H).
\end{align*}
In particular: 
\begin{equation} \IP^N_t(Z \in A) = \IP^N_0(^{N}\Phi_{t,0}(Z) \in A). \end{equation} 
\noindent We will often omit the index $N$ for $\IP_0$ and $\IE_0$ defined with respect to the product measure $\otimes^N f_0$. \\

\noindent To quantify the convergence of probability measures, we will use the Wasserstein distances (also known as Monge-Kantorivich-Rubinstein distances). In the context of kinetic equations, they were first introduced by Dobrushin in \cite{Dobrushin}. We shall briefly recall the definition and some basic properties. For further details, we refer the reader to the  book of Villani \cite[Ch. 6]{Villani}.

\begin{Definition}\label{Def_W}
	Let $\M(\IR^n)$ be the set of probability measures on $\IR^n$.	For given $\mu, \nu \in \M(\IR^n)$, let $\Pi(\mu, \nu)$ be the set of all probability measures $\IR^n \times \IR^n$ with marginal $\mu$ and $\nu$, respectively. Then, for $p \in [1, +\infty)$, the \emph{p'th Wasserstein distance} on $\M(\IR^n)$ is defined by
	\begin{equation}
	W_p(\mu, \nu) := \inf\limits_{\pi \in \Pi(\mu,\nu)} \, \Bigl( \int\limits_{\IR^n\times\IR^n} \lvert x - y \rvert^p \, \dd \pi(x,y) \, \Bigr)^{1/p}.
	\end{equation}
	\noindent Convergence in Wasserstein distance implies, in particular, weak convergence in $\M(\IR^n)$, i.e. 
	\begin{equation*} \int \Phi(x)\, \dd \mu_k(x) \to \int \Phi(x)\, \dd \mu(x), \;\;\; k \to \infty, \end{equation*}
	for all bonded, continuous functions $\Phi$. Moreover, convergence in $W_p$ implies convergence of the first $p$ moments. $W_p$ satisfies all properties of a metric on $\M(\IR^n)$, except that it may take the value $+\infty$. 

 The most common version is the first Wasserstein distance, for which we have the \emph{Kantorovich-Rubinstein duality}:
	\begin{equation}\label{Kantorovich} W_1(\mu, \nu) = \sup\limits_{\lVert g \rVert_{Lip} \leq 1}\Bigl\lbrace \int g(x) \, \dd \mu(x) -  \int g(x)\, \dd \nu(x) \Bigr\rbrace, \end{equation}
	where  $\lVert g \rVert_{Lip}:= \sup\limits_{x,y} \frac{g(x)-g(y)}{\lvert x-y \rvert}$, for $g:\IR^n\to \IR$.
	We will also consider the \emph{infinite Wasserstein distance} defined by
	
		\begin{equation}\label{InfiniteWasserstein} W_\infty(\mu, \nu) = \inf \lbrace \pi- \mathrm{esssup}\,  \lvert x - y \rvert \, : \, \pi \in \Pi(\mu,\nu)\rbrace.\end{equation}	
	
\end{Definition}

\noindent We can now state our precise results in the following theorem.

\begin{Theorem}[Molecular chaos]\label{Thm:Thm1}
	Let $f_0 \in L^\infty(\IR^3 \times \IR^3)$ a probability measure satisfying the assumptions of Theorem \ref{Thm:LP} a) and b) and $f^N$ the unique solution of the regularized Vlasov-Poisson equation \eqref{Vlasov} with initial datum $f_0$. For $0 < \delta < \frac{1}{3}$ let ${}^N\Psi_{t,s}$ be the $N$-particle flow solving \eqref{microscopiceq1} with cut-off width $N^{-\delta}$ and let ${}^N\Phi_{t,s}$ be the $N$-particle mean field flow induced by $f^N$ as defined in \eqref{meanfieldflow}. Then, for any $T>0$, there exists a constant $C_0$ depending on $\sup\limits_{N\in\IN}\lVert \rho[f^N] \rVert_{L^\infty([0,T]\times\IR^3)}$ such that for any $\beta > 0$ there exists a constant $C_\beta$ such that for all $N \geq N_0:= e^{\left(\frac{ C_0 T +1}{1 - 3\delta}\right)^2} $
	\begin{equation}\begin{split}\label{VPTHM}
	\IP_0\Bigl[ \exists t \in [0,T] : \lvert {}^N\Psi_{t,0}(Z) - {}^N\Phi_{t,0}(Z) \rvert_\infty \geq  N^{-\delta}  \Bigr]
	\leq  \frac{TC_\beta}{N^\beta},
	\end{split}\end{equation}
	where $\lvert \cdot \rvert_\infty$ denotes the maximum-norm on $\IR^{6N}$. 
\end{Theorem}

\noindent This result implies molecular chaos in the following sense: 
\begin{Corollary}
Let $F^N_0 := \otimes^N f_0$ and $F^N_t := {}^N\Psi_{t,0} \# F_0$ the $N$-particle distribution evolving with the microscopic flow \eqref{microNeq}. Then the $k$-particle marginal
			\begin{equation} 
			^{(k)}F^N_t (z_1, ..., z_k) := \int F^N_t(Z) \, \mathrm{d}^6z_{k+1}...\mathrm{d}^6z_N
			\end{equation}
			converges weakly to $\otimes^k f_t$
			{as} $N \to \infty$ for all $k \in  \IN$, where $f_t$ is the unique solution of the Vlasov-Poisson equation \eqref{VP} with $f^N\lvert_{t=0} = f_0$.  More precisely, under the assumptions of the previous theorem, we get  a constant $C >0$ such that for all $N \geq N_0$
			\begin{equation} 
			W_1(^{(k)}F_t^N, \otimes^k f_t) \leq k \, e^{TC\sqrt{\log(N)}}N^{-\delta},\, \forall 0 \leq t \leq T.
			\end{equation}
			 
\end{Corollary}
		
		\begin{proof} 
			
	 For fixed $0 \leq t \leq T$, let  $\mathcal{A} \subset \IR^{6N}$ be the set defined by $ Z \in \mathcal{A} \iff\bigl \lvert {}^N\Psi_{t,0}(Z) - {}^N\Phi_{t,0}(Z) \bigr \rvert_\infty < N^{-\delta}$. Hence, according to the previous theorem, $\IP_0(\mathcal{A}^c) \leq \frac{TC_\beta}{N^\beta}$ for sufficiently large $N$. In view of the Kantorovich-Rubinstein duality \eqref{Kantorovich}, we have:
			\begin{align*}\notag W_1({}^{(k)}F_t^N, &\otimes^k f_t)\\ 
			= &\sup\limits_{\lVert g \rVert_{Lip}=1}\, \Bigl\lvert \int \bigl( ^{(k)}F_t^N - \otimes^k f_t\bigr) g(z_1,...,z_k) \mathrm{d}^6z_1...\mathrm{d}^6z_k\Bigr\rvert  \notag\\\notag
			=& \sup\limits_{\lVert g \rVert_{Lip}=1}\, \Bigl\lvert\int \bigl(F_t^N (Z) -   \otimes^N f_t (Z) \bigr) g(z_1,...,z_k) \mathrm{d}^6z_1...\dd^6{z_k}...\mathrm{d}^6z_N\Bigr\rvert \\\notag
			= &\sup\limits_{\lVert g \rVert_{Lip}=1}\, \Bigl\lvert \int \bigl(\Psi_{t,0}\#F^N_0 (Z) -  \Phi_{t,0}\#F^N_0(Z)\bigr) g(z_1,...,z_k) \,\mathrm{d}^{6N}Z\Bigr\rvert
			\end{align*}
			Introducing the projection $P_k: \IR^{N} \to \IR^k, (z_1,...,z_N) \mapsto (z_1,...,z_k)$, this can be further rewritten as
			\begin{align}
			\notag W_1(&{}^{(k)}F_t^N, \otimes^k f_t)\\
			\notag
			= &\sup\limits_{\lVert g \rVert_{Lip}=1}\, \Bigl\lvert \int F^N_0(Z) \bigl(g(P_k \Psi_{t,0}(Z)) - g(P_k \Phi_{t,0}(Z))\bigr) \,\mathrm{d}^{6N}Z\Bigr\rvert\\
			=&\sup\limits_{\lVert g \rVert_{Lip}=1}\, \Bigl\lvert \int\limits_{\mathcal{A}^c} F^N_0(Z) \bigl(g(P_k \Psi_{t,0}(Z)) - g(P_k \Phi_{t,0}(Z))\bigr) \,\mathrm{d}^{6N}Z\Bigr\rvert \label{Acterm}\\\label{Aterm}
			+& \sup\limits_{\lVert g \rVert_{Lip}=1}\, \Bigl\lvert \int\limits_{\mathcal{A}} F^N_0(Z) \bigl(g(P_k \Psi_{t,0}(Z)) - g(P_k \Phi_{t,0}(Z))\bigr) \,\mathrm{d}^{6N}Z\Bigr\rvert.
			\end{align}

\noindent Using that all test-functions are Lipschitz with $\lVert g \rVert_{Lip} = 1$, we have $\eqref{Acterm} \leq  \IP_0(\mathcal{A}^c) \lVert F^N_0\rVert_\infty \lvert {}^N\Psi_{t,0}(Z) - {}^N\Phi_{t,0}(Z) \rvert$, with $\lVert F^N_0\rVert_\infty = (\lVert f_0 \rVert_\infty)^N$. Recalling that  $\lvert  {}^N\Psi_{0,0}(Z) - {}^N\Phi_{0,0}(Z) \rvert_\infty = 0$, we have
			\begin{align*}
			\lvert {}^N\Psi^2_{t,0}(Z) - {}^N\Phi^2_{t,0}(Z) \rvert_\infty &\leq \int\limits_{0}^t \lvert K^N_\delta(\Psi^1_{s,0}(Z)) - \overline{K}(\Phi^1_{s,0}(Z)) \rvert_\infty \dd s,\\
				\lvert {}^N\Psi^1_{t,0}(Z) - {}^N\Phi^1_{t,0}(Z) \rvert_\infty &\leq \int\limits_{0}^t  \lvert  {}^N\Psi^2_{s,0}(Z) - {}^N\Phi^2_{s,0}(Z) \rvert_\infty \dd s.
			\end{align*}

			\noindent The mean field force $\overline{K}$ is of order 1 (Lemma \ref{Lemma:forcebound}), while the microscopic force $K^N_\delta$ is bounded by $N^{2\delta}$. Hence, there exists a constant $C'>0$ such that $\lvert {}^N\Psi^2_{t,0}(Z) - {}^N\Phi^2_{t,0}(Z) \rvert_\infty \leq T C'N^{2\delta}$ and consequently $\lvert {}^N\Psi^1_{t,0}(Z) - {}^N\Phi^1_{t,0}(Z) \rvert_\infty \leq T^2C'N^{2\delta}$ for all $t \leq T$. Choosing $\beta := 3\delta$ in \eqref{VPTHM} we thus get another constant $C''$ such that
			\begin{equation}\label{actermbound} \eqref{Acterm} \leq C'' \max \lbrace T^2, T^3 \rbrace N^{-\delta}, \, \forall 0 \leq t \leq T.
			\end{equation}

				\noindent On the other hand, for $Z \in \mathcal{A}$, we have for any $g$ with $\lVert g \rVert_{Lip} = 1$, 
				\begin{equation*}  \lvert g(P_k {}^N\Psi_{t,0}(Z)) - g(P_k {}^N\Phi_{t,0}(Z)) \rvert \leq  \lvert  {}^N\Psi_{t,0}(Z) - {}^N\Phi_{t,0}(Z) \rvert_\infty \leq N^{-\delta}
				\end{equation*}
				
				\noindent for all $t \leq T$ and thus \eqref{Aterm} $\leq N^{-\delta}$. Together with \eqref{actermbound}, we get a constant $C'''$ such that
				\begin{equation}
				W_1(^{(k)}F_t^N, \otimes^k f^N_t) \leq C'''(1+T^3)N^{-\delta},\, \forall 0 \leq t \leq T.
				\end{equation}
		Finally, we will prove in Proposition \ref{Prop:fNtof} that 
		\begin{equation} W_1(f^N_t, f_t) \leq N^{-\delta} \,e^{t 2C_0 \sqrt{\log N}}, \forall t \leq T
		\end{equation}
where $C_0$ depends only on $f_0$ and $T$. Putting everything together and using $W_1(^{(k)}F_t^N, \otimes^k f_t) \leq W_1(^{(k)}F_t^N, \otimes^k f^N_t) + W_1(\otimes^k f^N_t, \otimes^k f_t)$ the statement follows. 
		\end{proof}
		
\noindent It is a classical result in probability theory (see e.g.  \cite{Kac}, \cite{Grunbaum}, \cite[Prop.2.2]{Sznitman}) that molecular chaos in the sense of the previous corollary implies convergence in law of the empirical distribution $\mu^N_t[Z]:=\mu^N_0[\Psi_{t,0}(Z)]$ to the constant variable $f_t$. However, under additional assumptions on the decay of $f_0$, we can obtain the following quantitative result:
		
\begin{Theorem}[Particle approximation of the Vlasov-Poisson system]\label{Thm:Thm}
	Let $f_0 \in L^\infty(\IR^3 \times \IR^3)$ a probability measure satisfying the assumptions of Theorem \ref{Thm:LP} a) and b). For $0 < \delta < \frac{1}{3}$, let ${}^N\Psi_{t,s}$ be the $N$-particle flow solving \eqref{microscopiceq1} with cut-off width $N^{-\delta}$. Let $p \in [1,\infty)$ and assume that, in addition, there exists $k>2p$ such that $\int_{\IR^{6}} \lvert z \rvert^k \dd f_0(z) < +\infty$. Then, the empirical density $\mu^N_t[Z]:= \mu^N_0[\Psi_{t,0}(Z)]$ converges to the solution of the Vlasov-Poisson equation in the following sense:\\
	
	\noindent For any $T>0$ and  $\gamma < \min \bigl \lbrace \frac{1}{6}, \frac{1}{2p}, \delta \bigr\rbrace$, there exists a constant $C_0$ depending on $f_0$ and $T$ and constants $c,C_1, C_2$ depending on $k, p, \gamma$ such that for all $N \geq N_1:= e^{\left(\frac{2(C_0 T + 1)}{1-3\delta}\right)^2}$
	
	\begin{equation}\begin{split}
	\IP_0\Bigl[ \exists t \in [0,T] : W_p(\mu^N_t[Z], f_t) >  N^{-\gamma +1 -3\delta}  \Bigr] \\
	\leq  C_1 e^{-cN^{1-(6 \vee 2p)\gamma}}  + T C_2 N^{1 - \frac{k}{2p}} ,
	\end{split}\end{equation}
	\noindent where $f$ is the unique solution of the Vlasov-Poisson system on $[0,T]$ with $f \lvert_{t=0} = f_0$ and  $6 \vee 2p := \max \lbrace 6, 2p \rbrace$. 
\end{Theorem}
\noindent The proofs of the theorems will be given in Sections \ref{section_proof} and \ref{section_proof2}.
\begin{Remarks}\mbox{}
	\begin{enumerate}[1)]
		
		\item Our results allow to choose the width of the cut-off arbitrary close to $N^{-1/3}$ which corresponds to the scale of the typical distance between a particle and its nearest neighbor.
		
		\item The results can be straightforwardly generalized to include external forces or noise, provided that the regularity of solutions to the corresponding mean field equation remains as assumed.
		
		\item All dynamical estimates required for Theorem \ref{Thm:Thm} would hold directly for the infinite Wasserstein distance. However, there are no concentration estimates available in terms of $W_\infty$ to establish a good approximation of the initial $f_0$ by the empirical measure~$\mu^N_0$. 
		
	\end{enumerate}
\end{Remarks}


\section{A new measure of chaos}\label{Section:measurechaos}

\noindent The strategy of proof is to control the deviation of the microscopic time evolution from the mean field time evolution in terms of the following $N$-dependent quantity: 

\begin{Definition}\label{Def:J}
	Let ${}^N\Phi_{t,0}$ the mean field flow defined in \eqref{meanfieldflow} and ${}^N\Psi_{t,0}$ the microscopic flow defined in \eqref{microscopicflow}. We denote by ${}^N\Phi^1_{t,0} = (q_i(t))_{1\leq i \leq N}$ and ${}^N\Phi^2_{t,0} = (p_i(t))_{1\leq i \leq N}$ the projection onto the spatial, respectively the momentum coordinates.\\
	
	\noindent 
	Consider the quantity $\Delta^N_{Z,t}$ defined as
	\begin{align*}\Delta^N_{Z,t}=\sqrt{\text{log}(N)} \lvert ^N\Psi^1_{t,0}(Z) - {}^N\Phi^1_{t,0}(Z) \rvert_\infty&\\+  \lvert ^N\Psi^2_{t,0}(Z) - {}^N\Phi^2_{t,0}(Z) \rvert_\infty&,\end{align*}
	where $\lvert Z \rvert_\infty = \max \lbrace \lvert z_i \rvert :  1\leq i\leq 3N \rbrace$ is the maximum-norm on $\IR^{3N}$.

	 \end{Definition}
\noindent In view of Theorem \ref{Thm:Thm1}, our aim is to show that for any $\epsilon > 0$:
 \begin{equation}\label{controlDelta} \IP_0\Bigl[\sup\limits_{0 \leq s \leq T}\left\{  \Delta^N_{Z,s} \right\}
 \geq N^{-\delta}   \Bigr] \to 0,
 \end{equation}
 faster than any inverse power of $N$. This will be done by introducing for $\lambda >0$ and $N \in \IN$  the stochastic process $J^{N,\lambda}_t(Z)$  given by
 \begin{equation}\label{eq:DefJ}J^{N,\lambda}_t(Z) := \min\left\{ 1,  \sup\limits_{0 \leq s \leq t}\left\{ e^{\lambda \sqrt{\text{log}(N)}(T-s)}\left(N^{\delta} \Delta^N_{Z,s}+N^{3\delta -1}	\right) \right\}  \right\}\end{equation}  and controlling the evolution of $\IE^N_0(J^{N,\lambda}_t)$.\\


\noindent The crucial innovation with respect to \cite{Peter} is that distances in spatial and momentum coordinates are weighted differently by a $N$-dependent factor (here: $\sqrt{\text{log}(N)}$), exploiting the second-order nature of the dynamics. Moreover, the quantity $J^{N,\lambda}_t$ has been redefined in order to optimize the rate of convergence. \\

\noindent The relevance of \eqref{controlDelta} for the proof of Theorem \ref{Thm:Thm} is grounded in the following observation:

\begin{Lemma}\label{Lemma:maxWinfty}
	For $X =(x_1,...,x_N) \in \IR^{6N}$ let $\mu^N_0[X] := \frac{1}{N} \sum\limits_{i=1}^N \delta_{x_i} \in \M(\IR^{6N})$. 
	Then we have for all $p \in [1,\infty]$:
	\begin{equation} W_p (\mu^N_0[X] , \mu^N_0[Y]) \leq \bigl\lvert X - Y\bigr\rvert_\infty.\end{equation}
\end{Lemma}

\noindent This implies, in particular, for any $\xi > 0$
\begin{align*}  
\IP_0 \Bigl[ \sup\limits_{0 \leq s \leq t} W_p(\mu^N_0[\Psi_{s,0}(Z)], \mu^N_0[\Phi_{s,0}(Z)]) \geq \xi\Bigr]& \\\leq \IP_0\Bigl[\sup\limits_{0 \leq s \leq T}\left\{  \Delta^N_{Z,s} \right\}
\geq \xi   \Bigr]&.\end{align*}

\begin{proof}[Proof of the Lemma]
	Since $W_p \leq W_q$ for $p\leq q$, it suffices to consider the infinite Wasserstein distance defined in \eqref{InfiniteWasserstein}.

	\noindent  We then observe that $\pi_0 =  \sum\limits_{i=1}^N \delta_{x_{i}}\delta_{y_{i}} \in \Pi(\mu^N_0[Z], \mu^N_0[Y])$ with  $ \pi_0- \mathrm{esssup}\,  \lvert x - y \rvert  = \max\limits_{1 \leq i \leq N} \lvert x_i - y_i\rvert  = \lvert X - Y \rvert_\infty$.
\end{proof}




\noindent In total, we will split our approximation result for the Vlasov-Poisson equation into
\begin{align}\label{1split} W_p(\mu^N_t[Z] , f_t ) &\leq W_p(\mu^N_0[\Psi_{t,0}(Z)], \mu^N_0[\Phi_{t,0}(Z)])\\\label{2split}
	&+ W_p(\mu^N_0[\Phi_{t,0}(Z)], f^N_t)\\\label{3split}
	&+ W_p(f^N_t, f_t). \end{align}

\noindent The first term \eqref{1split} is the interesting one, concerning the difference between microscopic time-evolution and mean field time-evolution. It will be controlled in terms of $\IE_0(J^{N,\lambda}_t)$ and shown to converge in probability faster than any inverse power of $N$.\\[1.5ex]
\noindent The second term $W_p(\mu^N_0[\Phi_{t,0}(Z)], f_t) = W_p(\varphi^N_{t,0} \# \mu^N_0[Z], \varphi^N_{t,0} \#  f_0)$ concerns the sampling of the mean field dynamics by discrete particle trajectories. We will use a recent large deviation estimate of Fournier and Guillin \cite{Fournier} to determine the typical rates of convergence for the initial distribution and then control the growth of \eqref{2split} by a Gronwall estimate.\\[1.5ex]
\noindent Convergence of \eqref{3split} is a purely deterministic result: solutions of the regularized Vlasov-Poisson equation \eqref{Vlasov} approximate solutions of the proper Vlasov-Poisson equation \eqref{VP} as the width of the cut-off goes to zero.\\


\noindent The central idea of our strategy is thus to first sample the (regularized) mean field dynamics along trajectories with random initial conditions, i.e. approximate $f^N_t$ by $\mu^N_0[\Phi_{t,0}(Z)]$, and then control the difference between mean field trajectories and the ``true'' microscopic trajectories in terms of the expectation $\IE^N_0(J^{N, \lambda}_t)$. This approach has several important virtues:

\begin{enumerate}
	\item The method is designed for stochastic initial conditions, thus allowing for law-of-large number estimates that turn out to be very powerful. (Note that the particles evolving with the mean field flow remain statistically independent at all times.)
	\item The metric $ \lvert ^N\Psi_{t,0}(Z) - {}^N\Phi_{t,0}(Z) \rvert_\infty$ is much stronger than usual weak distances between probability measures, thus allowing for better stability estimates. 
	
	\item Since $\frac{\dd}{\dd t} J^{N,\lambda}_t(Z) \leq 0$ if $ \sup\limits_{0 \leq s \leq t} \lvert ^N\Psi_{s,0}(Z) - {}^N\Phi_{s,0}(Z) \rvert_\infty \geq N^{-\delta}$ we only have to consider situations in which mean field trajectories and microscopic trajectories are still close together. 
	\item Exploiting the second-order nature of the dynamics, we weigh distances in $x$-space and momentum space differently, with an $N$-dependent factor $\sqrt{\text{log}(N)}$. As we compare microscopic trajectories to characteristic curves of the mean field equation, the growth the \textit{spatial} distance is trivially bounded by the difference of the respective momenta. The idea is thus to be a little more strict on deviations in space, so to speak, and use this to obtain better control on fluctuations of the force.
\end{enumerate}

\section{Local Lipschitz bound}
If all forces were Lipschitz continuous with a Lipschitz constant $L$ independent of $N$, we could easily conclude that $\frac{\dd}{\dd t} \lvert ^N\Psi_{t,0}(Z) - {}^N\Phi_{t,0}(Z) \rvert_\infty \leq (1+L) \lvert ^N\Psi_{t,0}(Z) - {}^N\Phi_{t,0}(Z) \rvert_\infty$, from which the desired convergence readily follows. However, the forces considered here become singular in the limit $N \to \infty$ and hence do not satisfy a uniform Lipschitz bound. Nevertheless, we observe that, for the mean field force $k^N_\delta * \rho^N_t$, the global Lipschitz constant $\lVert k^N * \rho^N_t  \rVert_{Lip}$ diverges only logarithmically as the cut-off is lifted with increasing $N$. Due to the pre-factor $\sqrt{\log{(N)}}$ in Definition \ref{Def:J}, the particular anisotropic scaling of our metric will allow us to ``trade'' part of this  divergence for a tighter control on spatial fluctuations. This will suffice to establish the desired convergence, using the fact that $e^{\sqrt{\log(N)}} = N^{\frac{1}{\sqrt{\log(N)}}}$ grows slower than $N^\epsilon$ for any $\epsilon >0$. (C.f. also \cite{Dustin} where we have implemented the same idea).\\

\noindent  We summarize our first observation in the following Lemma. 


\begin{Lemma}\label{lbound}
	Let $0<\gamma<1 $ and assume that $l : \IR^3 \to \IR^k$ satisfies
	\begin{equation} \lvert l(q) \rvert \leq c \cdot \min \lbrace N^{3\gamma} , \lvert q \rvert^{-3} \rbrace \end{equation}
	for some $c>0$. Then there exists a constant $C_l>0$ such that
	\begin{equation} \lVert l * \rho_t (x) \rVert_\infty \leq  C_{l} \max \lbrace 1, \sqrt{\log(N)} \rbrace \, \bigl(\lVert \rho_t \rVert_1 + \lVert \rho_t \rVert_\infty\bigr).
	\end{equation} 
\end{Lemma}
\begin{proof} 
We estimate	
	\begin{align*}  \lVert l * \rho_t (x) \rVert_\infty =&  \Bigl \lVert \int l(x-y) \rho_t (y) \, \mathrm{d}^3y \Bigr  \rVert_\infty \\
	\leq & \Bigl\lVert \int\limits_{\lvert x- y \rvert <N^{-\gamma}} l(x-y) \rho_t (y) \, \mathrm{d}^3y \Bigr\rVert_\infty \\
	+& \Bigl\lVert \int\limits_{N^{-\gamma}<\lvert x-y \rvert <1} l(x-y) \rho_t (y) \, \mathrm{d}^3y \Bigr \rVert_\infty\\
	+&  \Bigl\lVert \int\limits_{\lvert x-y \rvert >1} l(x-y) \rho_t (y) \, \mathrm{d}^3y \Bigr\rVert_\infty.
	\end{align*}
	
	\noindent The first term is bounded by
	\begin{align*} 
	\Bigl\lVert \int\limits_{\lvert x- y \rvert <N^{-\gamma}} l(x-y) \rho_t (y) \, \mathrm{d}^3y \Bigr\rVert_\infty
	\leq \lVert \rho_t \rVert_\infty N^{3\gamma} \lvert \mathrm{B}({N^{-\gamma}})  \rvert 
	\leq \frac{4}{3}\pi \, \lVert \rho_t \rVert_\infty,
	\end{align*}
	
	\noindent where $\mathrm{B}(r)$ denotes the ball with radius $r$. The last term is bounded by
	\begin{align*} 
	\Bigl\lVert \int\limits_{\lvert x-y \rvert >1} l(x-y) \rho_t (y) \, \mathrm{d}^3y \Bigr\rVert_\infty \leq c\, \lVert \rho_t \rVert_1.
	\end{align*}
	
	\noindent Finally, the second term yields
	
	\begin{align*} 
	\Bigl\lVert \int\limits_{N^{-\gamma}<\lvert x-y \rvert <1} g(x-y) \rho_t (y) \, \mathrm{d}^3y \Bigr\rVert_\infty &\leq \lVert \rho_t \rVert_\infty\int\limits_{N^{-\gamma}<\lvert y \rvert <1} \, \frac{c}{\lvert y \rvert^3} \, \mathrm{d}^3 y\\
	& \leq 4\pi c \gamma\, \lVert \rho_t \rVert_\infty \log(N).
	\end{align*}
\end{proof}

\noindent One immediate application of the Lemma is to $l(q)=\nabla k^N_\delta(q)$, showing that a regularized mean field force  is Lipschitz continuous with a constant proportional to $\log(N)$.  Our aim is now to prove that for typical initial conditions, fluctuations in the \emph{microscopic} forces can be bound in a similar fashion, as long as ${}^N\Psi_{t,0}(Z)$ and ${}^N\Phi_{t,0}(Z)$ are close.

\begin{Definition}\label{Def:L}
	Let
	\begin{equation} l^N_\delta(q) := \begin{cases} \frac{54}{\lvert q \rvert^3} &, \text{if}\, \lvert q \rvert \geq 3 N^{-\delta}\\ \, N^{3\delta} &, \text{else}\end{cases}  \end{equation}
	
	\noindent and $L: \IR^{6N} \to \IR^{N}$ be defined by $(L(Z))_i := \frac{1}{N} \sum\limits_{j \neq i} l^N_\delta(q_i - q_j)$. Furthermore,  we define $\overline{L}_t(Z)$ by $(\overline{L}_t(Z))_i := l^N_\delta * \rho^N_t(q_i) =  \int l^N_\delta*_q f(t, q_i, p) \,\dd^3 p $.
\end{Definition}

\begin{Lemma}\label{Lemma:Lipschitz} For any $\xi \in \IR^3$ with $\lvert \xi \rvert_\infty < 2 N^{-\delta}$, it holds that
	\begin{equation}\label{lipschitz} \lvert k^N_\delta(q) - k^N_\delta(q+\xi) \rvert_\infty \leq l^N_\delta(q) \lvert \xi \rvert_\infty.\end{equation}
\end{Lemma}

\begin{proof} First note that by assumption the derivative of $k^N$ is bounded by $N^{3\delta}$, so that \eqref{lipschitz} holds for $\lvert q \rvert < 3 N^{-\delta}$. For $\lvert q \rvert \geq 3 N^{-\delta}$, there exists $s \in [0,1]$ such that 
	$\lvert k^N_\delta(q) - k^N_\delta(q+\xi) \rvert  \leq \lvert \nabla k^N_\delta(q +s\xi)\rvert_\infty \lvert \xi \rvert_\infty$, where 
	\begin{equation} \lvert \nabla k^N_\delta(q +s\xi)\rvert_\infty \leq 2 \lvert q +s\xi \rvert^{-3}. \end{equation}
	The expression on the right-hand-side takes its greatest value if $\xi$ is antiparallel to $q$ and $s=1$. Hence, we have
	\begin{align}\lvert k^N_\delta(q) - k^N_\delta(q+\xi) \rvert_\infty  \leq 2\, \bigl\lvert q(1-\frac{\lvert \xi \rvert}{\lvert q \rvert}) \bigr\rvert^{-3} \, \lvert \xi \rvert_\infty.\end{align}
	
	\noindent Since $\lvert q \rvert \geq  3 N^{-\delta}$ and $\lvert \xi \rvert <2  N^{-\delta}$, it follows that $\frac{\lvert \xi \rvert}{\lvert q \rvert} < \frac{2}{3}$. Thus, we get
	$\lvert k^N_\delta(q) - k^N_\delta(q+\xi) \rvert_\infty \leq 2 \Bigl(\frac{3}{\lvert q \rvert} \Bigr)^3 \, \lvert \xi \rvert_\infty \leq \frac{54}{\lvert q \rvert^3} \, \lvert \xi \rvert_\infty$.
\end{proof}

\section{Law of large numbers}
In order to control the evolution of $\IE_0(J^N_t)$, we will require as an intermediate step that the mean field force (and its derivative) can be approximated by the analogous expressions for the discrete measure $\mu^N_0[\Phi_{t,0}(Z)]$ with random $Z$. The key observation here is that if the $N$-particle configuration evolves with the mean field flow ${}^N\Phi_{t,0}$, the particles remain statistically independent for all $t$, thus giving rise to a law-of-large-numbers estimate. 

\begin{Definition} \label{Def:sets}
	For any $t > 0$ and fixed $\delta < \frac{1}{3}$, we consider the ($N$ and $t$ dependent) sets $\mathcal{A}_t, \mathcal{B}_t, \mathcal{C}_t$ defined by
	\begin{align*}
	&Z \in \mathcal{A}_t \iff \lvert J^{N,\lambda}_t(Z) \rvert < 1\\
	&Z \in \mathcal{B}_t \iff \bigl\lvert K^N_\delta(\Phi_{t,0}(Z)) - \overline{K}(\Phi_{t,0}(Z)) \bigr\rvert_\infty < N^{- 1 + 2\delta}\\
	&Z \in \mathcal{C}_t \iff \bigl\lvert L^N_{\delta}(\Phi_{t,0}(Z)) - \overline{L}(\Phi_{t,0}(Z)) \bigr\rvert_\infty < 1
	\end{align*}
	\noindent where $\overline{K}$ is the mean field force \eqref{Kbar} and $\overline{L}$ as in Definition \ref{Def:L}. 
\end{Definition}

\noindent We now want to show that for any $t$, initial conditions in $\mathcal{B}_t \cap \mathcal{C}_t$ are \textit{typical} with respect to the product measure $F_0 := \otimes^N f_0$ on $\IR^{6N}$. 

\begin{Proposition}\label{Prop:LLN1} Let $\rho_t \in L^1(\IR^3)\cap L^\infty(\IR^3)$ with $\lVert \rho_t \rVert_1 =1$ as before. Let $h:\IR^3 \to \IR$ and suppose that for given $\delta > 0$ and $N\in \IN$ there exists $c > 0$ and an exponent $2 \leq \alpha \leq 3$ such that $\lvert h(x) \rvert \leq c\cdot \min \lbrace N^{\alpha \delta}, \lvert q \rvert^{-\alpha} \rbrace, \, \forall q \in \IR^3$.  Assume furthermore  that
	\begin{equation}\label{deltabound}
	\delta < \min \Bigl \lbrace  \frac{1-2\beta}{2\alpha -3} ,\; \frac{1-\beta}{\alpha} \Bigr \rbrace.
	\end{equation}
	Then there exists for all $\gamma > 0$  a constant $C_\kappa >0$ such that
	\begin{align} \IP_t\Bigl[\sup\limits_{1\leq i \leq N} \Bigl\lvert \frac{1}{N}\, \sum\limits_{j\neq i}^N h(q_i - q_j) - h*\rho_t (q_i) \Bigr\rvert \geq N^{-\beta}\Bigr] \leq \frac{C_\kappa}{N^\kappa}.
	\end{align}
	
\end{Proposition}

\begin{proof}
	Let 
	\begin{equation} D_i := \Bigl\lbrace Z \in \IR^6 :  \Bigl\lvert \frac{1}{N}\, \sum\limits_{j\neq i}^N h(q_i - q_j) - h*\rho_t (q_i) \Bigr\rvert \geq N^{-\beta} \Bigr\rbrace \end{equation} and $D := \bigcup\limits_{i=1}^N D_i$. Then $\IP(D)\leq \sum\limits_{i=1}^N \IP(D_i)= N \IP(D_1)$.\\
	
	\noindent By Markov's inequality, we have for every $M \in \IN$: 
	\begin{equation}\begin{split} \IP_t(D_1) \leq& \IE_t\Bigl[ N^{2M\beta}\, \Bigl\lvert \frac{1}{N}\, \sum\limits_{j=1}^N h(q_1 - q_j) - h*\rho_t (q_1) \Bigr\rvert^{2M} \Bigr]\\
	=&  \frac{1}{N^{2M(1-\beta)}}\,  \IE \Bigr[  \Bigl(\sum\limits_{j=1}^N \bigl( h (q_1 - q_j) - h*\rho_t (q_i)\bigr) \Bigr)^{2M}\Bigr].
	\end{split}\end{equation}
	
	\noindent Let $\mathcal{M} := \lbrace \mathbf{k} \in \IN_0^N \mid \lvert \mathbf{k} \rvert = 2M \rbrace$ the set of multiindices $\mathbf{k} = (k_1, ..., k_N)$ with $\sum\limits_{j=1}^{N} k_j = 2M$. Let
	\begin{equation} G^\mathbf{k} := \prod \limits_{j=1}^N \bigl( h (q_j - q_1) - h*\rho_t (q_1)\bigr)^{k_j}. \end{equation}
	\noindent Then: 
	\begin{equation} \IE_t \Bigr[ \Bigl(\sum\limits_{j=1}^N \bigl( h (q_1 - q_j) - h*\rho_t (q_1)\bigr) \Bigr)^{2M}\Bigr] =  \sum\limits_{\mathbf{k} \in \mathcal{M}}\binom{2M}{\mathbf{k}} \, \IE_t(G^\mathbf{k}). \end{equation}
	
	\noindent Now we note that  $\IE_t(G^\mathbf{k}) = 0$ whenever there exists a $1 \leq j \leq N$ such that $k_j =1$. This can be seen by integrating the j'th variable first.\\
	
	\noindent For the remaining  terms, we have for any $1 \leq m \leq M$:
	\begin{align*} \int \lvert h(q_1-q_j) \rvert^m f_t(q_j, p_j) \, \mathrm{d}^3p_j\, \mathrm{d}^3p_j =  \int \lvert h\rvert^m (q_1-q_j)\rho_t(q_j) \, \mathrm{d}^3q_j.\end{align*}
	Now for $2 \leq \alpha < 3$ and $m=1$ we estimate
	\begin{align*} &\lvert h*\rho_t(q_1)\rvert \leq \int \lvert h\rvert (q_1-y)\rho_t(y) \, \mathrm{d}^3y\\\notag
	& \leq c \int\limits_{\lvert y\rvert < 1} \lvert y \rvert^{-\alpha}\, \rho_t(q_1-y) \, \dd^3 y + c \int\limits_{\lvert y\rvert \geq 1} \lvert q_j \rvert^{-\alpha } \rho_t(q_1-y) \, \dd^3 y \\
	& \leq c\, \bigl( 4 \pi \lVert \rho_t \rVert_\infty + \lVert \rho_t \rVert_1 \bigr),
	\end{align*}
	while  for $\alpha = 3$, we find:
	\begin{align*}
	&\lvert h*\rho_t(q_1)\rvert \leq \int \lvert h\rvert (q_1-y)\rho_t(y) \, \mathrm{d}^3y\\\notag
	& \leq c \Bigl( \int\limits_{\lvert y\rvert \leq N^{-\delta}} + \int\limits_{ N^{-\delta} < \lvert y\rvert < 1}  +   \int\limits_{\lvert y\rvert \geq 1} \Bigr) \lvert h(y) \rvert\, \rho_t(q_1-y) \, \dd^3 y\\
	&\leq c \lVert \rho_t \rVert_\infty \int\limits_{\lvert y\rvert \leq N^{-\delta}} N^{3\delta} \, \dd^3 y + c \lVert \rho_t \rVert_\infty \int\limits_{ N^{-\delta} < \lvert y\rvert < 1} \frac{1}{\lvert y \rvert^3} \dd^3 y +  c \int\limits_{\lvert y\rvert \geq 1} \rho_t(q_1-y) \, \dd^3 y\\
	&\leq c \, \Bigl(4\pi \lVert \rho_t \rVert_\infty(\frac{1}{3} + \log(N^\delta)) + \lVert \rho_t \rVert_1 \Bigr).
	\end{align*}
	
	\noindent For $m \geq  2$, we find in any case 
	\begin{align*}
	&\int \lvert h\rvert^m (q_1-y)\rho_t(y) \, \mathrm{d}^3y = \int \lvert h\rvert^m (y)\rho_t(q_1-y) \, \mathrm{d}^3y\\
	\leq &\int\limits_{\lvert y \rvert <N^{-\delta}} \lvert h \rvert^m (y) \rho_t(q_1-y) \, \mathrm{d}^3y +  \int\limits_{\lvert y \rvert \geq N^{-\delta}} \lvert h \rvert^m (y) \rho_t(q_1-y) \, \mathrm{d}^3y \\
	\leq &c \lVert \rho_t \rVert_\infty \Bigl( 4\pi N^{-3\delta} N^{\alpha \delta m} +  \int\limits_{\lvert y  \rvert \geq N^{-\delta}} \frac{1}{\lvert y \rvert^{\alpha m}} \, \mathrm{d}^3 y \Bigl)
	\leq 8\pi c \lVert \rho_t\rVert_\infty\;  N^{(\alpha m - 3)\delta}.
	\end{align*}

	\noindent  Hence, setting $C_\alpha := 16\pi c \lVert \rho_t \rVert_\infty \bigl(1 + \mathds{1}_{\lbrace \alpha = 3\rbrace} \log(N) \bigr)$ we can conclude that for all $m \geq 2$:
	\begin{equation}\bigl \lvert h (q_j - q_i) - h (q_i)\bigr \rvert^{m} \\
	\leq  C^m_\alpha N^{(\alpha m - 3)\delta}.
	\end{equation}
	
	\noindent Now, for $\mathbf{k}=(k_1, k_2, ..., k_N) \in \mathcal{M}$, let $\# \mathbf{k}$ denote the number of $k_j$ with $\alpha k_j \neq 0$. Note that if $\# \mathbf{k} > M$, we must have $k_j =1$ for at least one $1 \leq j \leq N$, so that $\mathbb{E}_t(G^\mathbf{k}) = 0$. For the other multiindices, we get (using that the particles are statistically independent):
	\begin{equation}\begin{split}\label{LLNestimate} \mathbb{E}_t (G^\mathbf{k}) = \mathbb{E}_t \Bigl[ &\prod \limits_{j=1}^N \bigl( k_\delta (q_j - q_i) - k*\rho_t (q_i)\bigr)^{k_j}\Bigr]\\
	\leq  &\prod \limits_{j=1}^N  \mathbb{E}_t \Bigl[\bigl(\lvert h (q_j - q_i)\rvert + \lvert h*\rho_t (q_i)\rvert \bigr)^{k_j}\Bigr]\\
	\leq  &\prod \limits_{j=1}^N C_\alpha^{k_j} \,  N^{(\alpha k_j -3)\delta}\\
	\leq & C_\alpha^{2M} N^{2M\alpha\delta}  N^{-3\delta \#\mathbf{k}}.
	\end{split}
	\end{equation}
	\noindent Finally, we observe that for any $l \geq 1$, the number of multiindices $\mathbf{k} \in \mathcal{M} $ with $\#\mathbf{k} = l$ is bounded by
	\begin{equation*} \sum\limits_{\#\mathbf{k} = l} 1 \leq \binom{N}{l} (2M)^l  \leq (2M)^{2M} N^l. \end{equation*}
	
	\noindent Thus:
	\begin{align*}\notag \IP_t(D_1) &\leq \frac{1}{N^{2M(1-\beta)}} \sum\limits_{\mathbf{k} \in \mathcal{M}} \binom{2M}{\mathbf{k}}\, \IE_t(G^\mathbf{k})\\\label{Mterm}
	&\leq C_\alpha^{2M} C_M  \, \frac{N^{2M\alpha\delta}}{N^{2M(1-\beta)}}\, \sum\limits_{l=1}^{M} N^{(1-3\delta)l}\\
	&\leq C_\alpha^{2M} M C_M N^{2M(\alpha\delta + \beta -1)}\, \max\lbrace N^{M(1-3\delta)}, 1 \rbrace\\
	&\leq C_\alpha^{2M} M C_M N^{-\epsilon M},
	\end{align*}
	where $C_M$ is some constant depending on $M$ and 
	\begin{equation} \epsilon := \begin{cases} 1- 2\beta + \delta (3-2\alpha) & \text{ if } 3\delta < 1\\
	2 (1-\beta - \alpha\delta) & \text{ if } 3\delta \geq 1.
	\end{cases}
	\end{equation}  
	$\epsilon \geq 0$ according to \eqref{deltabound}. For $2 \leq \alpha < 3$ we conclude the proof by noting that
	\begin{align}\label{LLNbound} \IP_t(D) \leq N \,\IP_t(D_1) \leq C_\alpha^{2M} MC_M \, N^{-(\epsilon M +1)},\end{align}
	and choosing $M$ so large that $(\epsilon M - 1) = \gamma$. For $\alpha =3$, however, equation \eqref{LLNbound} becomes
	\begin{equation} \IP_t(D) \leq C'(M) (1 + \log(N))^{2M} N^{-(\epsilon M - 1)}, \end{equation} 
	where $C'(M)$ is some constant depending on $M$ and $\lVert \rho_t \rVert_\infty$. This can be rewritten as
	\begin{equation} (1 + \log(N))^{2M} N^{-\epsilon M + 1} = \Bigl(\frac{1 + \log(N)}{N^{\epsilon/4}} \Bigl)^{2M} N^{- \frac{\epsilon}{2}M + 1}. \end{equation}
	The function $g(x)=\frac{1 + \log(x)}{x^{\epsilon/4}}, \; x \in [1,\infty)$ is continuous with $\lim\limits_{x \to \infty} g(x) = 0$. Hence, it has a maximum $C < +\infty$. In particular, $\frac{1 + \log(N)}{N^{\epsilon/4}} \leq C$ independent of $N$ and the announced result holds for $\alpha =3$, as well.
\end{proof}

\begin{Corollary} Let $\mathcal{B}_t, \mathcal{C}_t$ as in Definition \ref{Def:sets}. Then we find for any  $\gamma > 0$ a constant $C_\kappa $ such that
	\begin{align*}
	\IP^N_0(\mathcal{B}_t) \geq 1 -  \frac{C_\kappa}{N^\kappa},\\
	\IP^N_0(\mathcal{C}_t) \geq 1- \frac{C_\kappa}{N^\kappa}.
	\end{align*}
	In other words, for any fixed $t$, initial conditions in $\mathcal{B}_t \cap \mathcal{C}_t$ are typical with the measure of ``bad'' initial conditions decreasing faster than any inverse power of $N$.  
\end{Corollary}
\begin{proof}
Note that 
\begin{align*} Z \in {}^N\Phi_{t,0}(\mathcal{B}_t) &\iff  \bigl\lvert K^N_\delta(Z) - \overline{K}(Z) \bigr\rvert_\infty < N^{- 1 + 2\delta}\\
	&\iff \max\limits_{1\leq i \leq N} \Bigl\lvert \frac{1}{N}\, \sum\limits_{j\neq i}^N k^N_\delta (q_i - q_j) - k^N_\delta*\rho^N_t (q_i) \Bigr\rvert < N^{-1+2\delta}
\end{align*} 
and similarly
\begin{align*} 
	Z \in {}^N\Phi_{t,0}(\mathcal{C}_t) &\iff \bigl\lvert L^N_\delta(Z) - \overline{L}(Z) \bigr\rvert_\infty < 1\\
	&\iff \max\limits_{1\leq i \leq N} \Bigl\lvert \frac{1}{N}\, \sum\limits_{j\neq i}^N l^N_\delta (q_i - q_j) - l^N_\delta*\rho^N_t (q_i) \Bigr\rvert < 1.
\end{align*}
		Applying the previous result once for $k^N_\delta$ with $\alpha =2$ and $\beta=1-2\delta$ and once for $l^N_\delta$ with $\alpha=3$ and $\beta = 0$, we get
		\begin{align} &\IP^N_t\Bigl[\max\limits_{1\leq i \leq N} \Bigl\lvert \frac{1}{N}\, \sum\limits_{j\neq i}^N k^N_\delta (q_i - q_j) - k^N_\delta*\rho^N_t (q_i) \Bigr\rvert \geq N^{-1+2\delta}\Bigr] \leq \frac{C_\gamma}{N^\gamma},\\\label{LLNg}
			&\IP^N_t\Bigl[\max\limits_{1\leq i \leq N} \Bigl\lvert \frac{1}{N}\, \sum\limits_{j\neq i}^N l^N_\delta (q_i - q_j) - l^N_\delta*\rho^N_t (q_i) \Bigr\rvert \geq 1 \Bigr]\leq \frac{C_\gamma}{N^\gamma}.
		\end{align}
		Observing that $\IP^N_0(\mathcal{B}_t)=\IP^N_t(\Phi_{t,0}(\mathcal{B}_t))$ and $\IP^N_0(\mathcal{C}_t)=\IP^N_t(\Phi_{t,0}(\mathcal{C}_t))$, the statement follows. 
	\end{proof}


\section{Proof of Theorem \ref{Thm:Thm1}}\label{section_proof}
We now have everything in place to prove our first Theorem \ref{Thm:Thm1}. That is, we will prove the following statement: \\
 
\noindent \emph{Under the assumptions of Theorem \ref{Thm:Thm1}, we find for all $\delta < \frac{1}{3}$ and $T>0$ a constant $C_0 > 0$ such that for any $\beta>0$ there exists a constant $C >0$ such that for $N \geq  e^{(\frac{C_0 T +1}{1 - 3\delta})^2} $}
	\begin{equation} \IP_0 \Bigl[	\sup\limits_{0 \leq s \leq T}\left\{  \Delta^N_{Z,s} \right\}
						\geq  N^{-\delta}  \Bigr]\leq T \frac{C}{N^{\beta}}.
	\end{equation}
	

\noindent In order to control the evolution of $\sup\limits_{0 \leq s \leq t}\left\{  \Delta^N_{Z,s} \right \} $, respectively  $J^{N,\lambda}_t(Z)$ defined in \eqref{eq:DefJ}, we will need the following Lemma.

\begin{Lemma}\label{Lemma:semiderivative}
	For a function $g: \IR \to \IR$, we denote by
	\begin{equation} 
	\partial_t^+ g(t) := \lim\limits_{\Delta^N t \searrow 0} \frac{g(t + \Delta^N t) - g(t)}{\Delta^N t} 
	\end{equation}
	the right-derivative of $f$ with respect to $t$. Let $g \in C^1(\IR)$ and $h(t):= \sup\limits_{0 \leq s \leq t} g(s)$. Then $\partial_t^+ h(t)$ exists and $\partial_t^+ h(t) \leq \max \lbrace 0 , g'(t) \rbrace$ for all $t$. 
\end{Lemma}
\begin{proof} 
	We have to distinguish 3 cases.\\
	1) If $g'(t) \leq 0$, there exists $\Delta^N t > 0$ such that $g(s) \leq g(t), \forall s \in [t, t + \Delta^N t)$. Thus for all $t' \in [t, t + \Delta^N t)$ we have $h(t'):= \sup\limits_{0 \leq s \leq t'} g(s) =  \sup\limits_{0 \leq s \leq t} g(s) = h(t)$ and $\partial_t^+ h(t) = 0$.\\
	2) If $g(t) < h(t)$, there exists $\Delta^N t > 0$ such that $g(s) \leq h(t) \, \forall s \in (t- \Delta^N t, t + \Delta^N t)$. This means that $h$ is constant on $(t- \Delta^N t, t + \Delta^N t)$ so that, in particular, $\partial_t^+ h(t) = 0$.\\
	3) If $g(t) =  h(t)$ and $g'(t) > 0$, there exists $\Delta^N t > 0$ such that $g$ is monotonously increasing on $(t- \Delta^N t, t + \Delta^N t)$. Hence, we have $h(t') = \sup\limits_{0 \leq s \leq t'} g(s) = g(t')$ for all $t' \in [t + \Delta^N t)$ and thus $\partial_t^+ h(t) = g'(t)$. 
\end{proof}

\begin{proof}[\textbf{\em{Proof of Theorem \ref{Thm:Thm1}}}.]
	Recall from Definition \ref{Def:J}	
		\begin{align*}\Delta^N_{Z,t}= \sqrt{\text{log}(N)} \lvert ^N\Psi^1_{t,0}(Z) - {}^N\Phi^1_{t,0}(Z) \rvert_\infty& \\
		+  \lvert ^N\Psi^2_{t,0}(Z) - {}^N\Phi^2_{t,0}(Z) \rvert_\infty&.\end{align*}
As announced, we introduce for any $\lambda> 0$ and $N \in \IN$ the process
		$$J^{N,\lambda}_t(Z) := \min\left\{ 1,  \sup\limits_{0 \leq s \leq t}\left\{ e^{\lambda \sqrt{\text{log}(N)}(T-s)}\left(N^{\delta} \Delta^N_{Z,s}+N^{3\delta -1}
			\right) \right\}  \right\}.$$
	We consider the expectation $\IE_0(J^{N,\lambda}_t)$ which we split as follows:
	\begin{equation*} \IE_0(J^{N,\lambda}_t) = \IE_0(J^{N,\lambda}_t \mid \mathcal{A}_t^c) + \IE_0(J^{N,\lambda}_t \mid \mathcal{A}_t\setminus \mathcal{B}_t\cap\mathcal{C}_t) + \IE_0(J^{N,\lambda}_t \mid \mathcal{A}_t\cap\mathcal{B}_t\cap\mathcal{C}_t). \end{equation*}
	Here, $\mathcal{A}_t, \mathcal{B}_t,\mathcal{C}_t$ are the sets defined in Definition \ref{Def:sets} and $J^{N,\lambda}_t \mid \mathcal{A}_t$ denotes the restriction of $J^{N,\lambda}_t$ to the set $\mathcal{A}_t \subset \IR^{6N}$ etc.\\
	
	\noindent 1) On $\mathcal{A}_t^c$, we have $\frac{\mathrm{d}}{\mathrm{d}t} J^{N,\lambda}_t(Z) = 0 $, since $J^{N,\lambda}_t(Z)$ is already maximal and thus also
	\begin{equation}\label{eq:outsideA} \frac{\mathrm{d}}{\mathrm{d}t}\, \IE_t(J^{N,\lambda}_t \mid \mathcal{A}_t^c) = 0. \end{equation}
	
	\noindent 2) For $Z \in \mathcal{A}_t$, we have to consider $\partial_t^+ J_t^{N, \lambda}(Z) \leq \max \Bigl\lbrace 0 , I_t^{N, \lambda}(Z)  \Bigr \rbrace $
	with
	\begin{equation}\begin{split}\label{ddelta}
	I_t^{N, \lambda}(Z) := &\frac{d}{dt}\left(e^{\lambda \sqrt{\text{log}(N)}(T-t)}\left(N^{\delta} \Delta^N_{Z,t}+N^{3\delta -1}
	\right) \right)
	\\ = &-\lambda \sqrt{\text{log}(N)} e^{\lambda \sqrt{\text{log}(N)}(T-t)}\left(N^{\delta} \Delta^N_{Z,t}+N^{3\delta -1}\right)\\
&+  e^{\lambda \sqrt{\text{log}(N)}(T-t)} N^{\delta} \partial_t \Delta^N_{Z,t}.
	\end{split}\end{equation}

\noindent We split $\partial_t\Delta^N_{Z,t}$ into
	\begin{equation}\begin{split}\label{controlQ} &\partial_t \lvert {}^N\Psi^1 _{t,0}(Z) - {}^N\Phi^1_{t,0}(Z) \rvert_\infty \leq \lvert \partial_t ({}^N\Psi^1_{t,0}(Z) - {}^N\Phi^1_{t,0}(Z)) \rvert_\infty\\
	&\leq \lvert {}^N\Psi^2 _{t,0}(Z) - {}^N\Phi^2_{t,0}(Z) \rvert_\infty \leq \sup\limits_{0 \leq s \leq t} \lvert {}^N\Psi^2 _{s,0}(Z) - {}^N\Phi^2_{s,0}(Z) \rvert_\infty \end{split}\end{equation}
	and 
	\begin{equation}\begin{split}\partial_t \lvert {}^N\Psi^2 _{t,0}(Z) - {}^N\Phi^2_{t,0}(Z) \rvert_\infty &\leq  \lvert \partial_t ({}^N\Psi^2 _{t,0}(Z) - {}^N\Phi^2_{t,0}(Z) )\rvert_\infty\\ &\leq  \lvert K^N_\delta(\Psi^1_{t,0}(Z)) - \overline{K}_t(\Phi^1_{t,0}(Z)) \rvert_\infty.
	\end{split}\end{equation} 
	
	\noindent We begin by controlling the contribution of ``bad'' initial conditions not contained in $\mathcal{B}_t$ and $\mathcal{C}_t$. Since $\lvert k^N_\delta\rvert \leq N^{2\delta}$, the total force acting on each particle (with the $1/N$-scaling) is also bounded as $\lvert K^N_\delta(Z) \rvert_\infty \leq N^{2\delta}$. The mean field force $\overline{K}$ is of order $1$, according to Lemma \ref{Lemma:forcebound} and $N^{\delta} \lvert {}^N\Psi^2_{t,0}(Z) - {}^N\Phi^2_{t,0}(Z) \rvert_\infty  \leq 1$ since $Z \in \mathcal{A}_t$. In total, we thus have $\sup\lbrace \lvert \partial^+_t J^{N,\lambda}_t(Z) \rvert : Z \in \mathcal{A}_t \rbrace \leq C e^{\lambda \sqrt{\text{log}(N)}T} N^{3\delta}$ for some $C >0$. According to Prop. \ref{Prop:LLN1}, the probability for $Z \notin  \mathcal{B}_t \cap \mathcal{C}_t$ decreases faster than any power of $N$. Hence, we can find for any $\kappa > 0$ a constant $C_\kappa$ (which may differ from the $C_\kappa$ in  Proposition  \ref{Prop:LLN1}), such that
	\begin{equation}\begin{split}\label{controlatypical} &\partial_t^+ \IE_0(J^{N,\lambda}_t \mid \mathcal{A}_t\setminus(\mathcal{B}_t\cap\mathcal{C}_t))\\ &\leq \sup\lbrace \lvert \partial^+_t J^{N,\lambda}_t(Z) \rvert : Z \in \mathcal{A}_t \rbrace\, \IP_0\bigl[(\mathcal{A}_t \cap \mathcal{B}_t)^c\bigr]\\ &\leq e^{\lambda \sqrt{\text{log}(N)}T} \frac{C_\kappa}{N^\kappa}. \end{split}\end{equation}

	\noindent 3) It remains to control the change of $J^{N,\lambda}_t$ for the typical initial conditions $Z \in \mathcal{A}_t \cap \mathcal{B}_t \cap \mathcal{C}_t$. To this end, we consider:
	\begin{align} \notag &\lvert K^N_\delta(\Psi^1_{t,0}(Z)) - \overline{K}_t(\Phi^1_{t,0}(Z)) \rvert_\infty\\\label{1term}
	 &\leq \lvert K^N_\delta(\Psi^1_{t,0}(Z)) - K^N_\delta(\Phi^1_{t,0}(Z)) \rvert_\infty \\\label{2term}
	&+  \lvert K^N_\delta(\Phi^1_{t,0}(Z)) - \overline{K}_t(\Phi^1_{t,0}(Z)) \rvert_\infty.  \end{align}
	
	\noindent Since $Z \in \mathcal{B}_t$, it follows that 
	\begin{equation}\lvert K^N_\delta(\Phi^1_{t,0}(Z)) - \overline{K}_t(\Phi^1_{t,0}(Z)) \rvert_\infty < N^{2\delta -1}.\end{equation} 
For \eqref{1term}, we use the triangle to get for any $1 \leq i \leq N$: 
	\begin{align*}  \Bigl \lvert \bigl(K^N_\delta(\Psi^1_{t,0}(Z)) - K(\Phi^1_{t,0}(Z)) \bigr)_i \Bigr\rvert_\infty \leq \Bigl \lvert \sum\limits_{j=1}^N k^N_\delta(\Psi^1_i - \Psi^1_j) - k^N_\delta(\Phi^1_i - \Phi^1_j) \Bigr\rvert_\infty \\
	\leq \sum\limits_{j=1}^N \bigl \lvert k^N_\delta(\Psi^1_i - \Psi^1_j) - k^N_\delta(\Phi^1_i - \Phi^1_j) \bigr\rvert_\infty.
	\end{align*}

	\noindent Since $Z \in \mathcal{A}_t$, i.e. $J^{N,\lambda}_t(Z)<1$, we have in particular for $N>3$ i.e. $\log(N) >1$, $\sup\limits_{0 \leq s \leq t} \lvert ^N\Psi^1_{s,0}(Z) - {}^N\Phi^1_{s,0}(Z) \rvert_\infty< N^{-\delta}$. Therefore we can use Lemma \ref{Lemma:Lipschitz} to get the bound 
	\begin{align*}
	\bigl \lvert k^N_\delta(\Psi^1_i - \Psi^1_j) - k^N_\delta(\Phi^1_i - \Phi^1_j) \bigr\rvert_\infty &\leq l^N_{\delta}(\Phi^1_i - \Phi^1_j) \lvert (\Psi^1_i - \Psi^1_j) - (\Phi^1_i - \Phi^1_j) \rvert_\infty\\[1.2ex]
	&\leq 2\, l^N_{\delta}(\Phi^1_i - \Phi^1_j) \lvert \Psi^1_{t,0} - \Phi^1_{t,0} \rvert_\infty.
	\end{align*}
	
	\noindent Since $Z \in \mathcal{C}_t$, it follows with Lemma \ref{lbound} that 
	\begin{align*}  \sum\limits_{j=1}^N l^N_{\delta} (\Phi^1_i - \Phi^1_j) = \bigl(L^N_{\delta}(\Phi_{t,0}(Z) \bigr)_i \leq \lVert l^N_{\delta}*\rho^N_t(q) \rVert_\infty + 1 \leq 2C_l \log(N) (1+\lVert \rho^N_t \rVert_\infty). \end{align*}
 Hence, setting $C_0:= 2C_l C_\rho$, where $C_\rho$ is the uniform bound on the charge densities from \eqref{Crho}, we have
	\begin{align}\label{controlP} \frac{\mathrm{d}}{\mathrm{d}t}  \lvert \Psi^2 _t(Z) - \Phi^2_{t,0}(Z) \rvert_\infty \leq C_0  \log(N)\,  \bigl\lvert \Psi^1_{t,0}(Z) - \Phi^1_{t,0}(Z) \bigr\rvert_\infty + N^{2\delta -1} \end{align}
	
	\noindent for $Z \in \mathcal{A}_t \cap \mathcal{B}_t \cap \mathcal{C}_t$. Together with \eqref{controlQ}, this yields: 
	\begin{align*} \partial_t^+ \Delta^N_{Z,t}\Bigl\lvert_{\mathcal{A}_t \cap \mathcal{B}_t \cap \mathcal{C}_t}
	\leq\,  &\sqrt{\text{log}(N)} \frac{\mathrm{d}}{\mathrm{d}t}  \lvert \Psi^1_{t,0}(Z) - \Phi^1_{t,0}(Z) \rvert_\infty +    \frac{\mathrm{d}}{\mathrm{d}t}  \lvert \Psi^2_{t,0}(Z) - \Phi^2_{t,0}(Z) \rvert_\infty\\[1.2ex]
	\leq\,  &\sqrt{\text{log}(N)} \lvert \Psi^2_{t,0}(Z) - \Phi^2_{t,0}(Z) \rvert_\infty\\[1.2ex]
	 + \, & \Bigl[C_0 \log(N) \lvert \Psi^1_{t,0}(Z) - \Phi^1_{t,0}(Z) \rvert_\infty + N^{2\delta -1} \Bigr]
	\\ \leq \, &\; C_0\,\sqrt{\log{N}} \Delta^N_{Z,t} + N^{2\delta -1}.
	\end{align*}
	
	\noindent   Plugging this into equation \eqref{ddelta} and using that $C_l>4\pi$ we have found $\partial_t^+ J_t^{N, \lambda} (Z) \leq \max \lbrace 0 , I_t^{N, \lambda} (Z) \rbrace$ with
	\begin{align*}\notag 
		I_t^{N, \lambda} (Z)	&\leq
			 -\lambda \sqrt{\text{log}(N)} e^{\lambda \sqrt{\text{log}(N)}(T-t)}\left(N^{\delta} \Delta^N_{Z,t}+N^{3\delta -1}\right)
			\\\notag & \; \;\; + e^{\lambda \sqrt{\text{log}(N)}(T-t)}N^{\delta}\left( C_0 \sqrt{\log{N}} \Delta^N_{Z,t} + N^{2\delta - 1}\right)
			\\\notag &= \sqrt{\text{log}(N)}N^{\delta} e^{\lambda \sqrt{\text{log}(N)}(T-t)}
			\\  &\;\;\;\left[\left(C_0\, -\lambda \right)\Delta^N_{Z,t}+ \left(\frac{1}{\sqrt{\log(N)}}-\lambda \right) \, N^{2\delta -1}\right].
	\end{align*}
	
	\noindent Choosing $\lambda = C_0$ gives that this is negative. 
 Hence, we have
		\begin{equation}  \partial_t^+  \IE_0(J^{N,\lambda}_t \mid \mathcal{A}_t \cap \mathcal{B}_t \cap \mathcal{C}_t) = 0. \end{equation}

\noindent Together with \eqref{eq:outsideA} and \eqref{controlatypical} we have found: 
	$$ \partial_t^+ \IE_0(J^{N,\lambda}_t )  \leq e^{\lambda \sqrt{\log(N)}T}\frac{C_\kappa}{N^\kappa} $$
and thus \begin{align}\label{Ebound} &\IE_0(J^{N,\lambda}_t) -  \IE_0(J^{N,\lambda}_0)=   \IE_0 \bigl(J^{N,\lambda}_t-J^{N,\lambda}_0 \bigr) \leq T e^{\lambda \sqrt{\text{log}(N)}T} \frac{C_\kappa}{N^\kappa}, \end{align}
uniform in $t\in[0,T]$. Observing that $\Delta^N_{Z,0}=0$, we have $J^{N,\lambda}_0(Z)\equiv e^{\lambda \sqrt{\text{log}(N)}T} N^{3\delta -1}$. 
 Now we shall demand 
 \begin{equation}	\label{Nzero}
 N \geq N_0 := e^{(\frac{\lambda T + 1}{1 - 3\delta} )^2} \Rightarrow
 e^{\lambda \sqrt{\text{log}(N)}T} N^{3\delta -1} \leq \frac{1}{2}.\end{equation} 
(Here we exploited the fact that $e^{\sqrt{\log(N)}}$ grows slower than any power of $N$). The random variable $J^{N,\lambda}_t - J^{N,\lambda}_0= J^{N,\lambda}_t-e^{\lambda \sqrt{\text{log}(N)}T} N^{3\delta -1}$ is then certainly non-negative and it follows from \eqref{Ebound} that:	  
	\begin{equation*}\IP_0 \Bigl[J^{N,\lambda}_T(Z)-J^{N,\lambda}_0(Z) \geq \frac{1}{2} \Bigr] \leq 2T e^{\lambda \sqrt{\text{log}(N)}T} \frac{C_\kappa}{N^{\gamma} }. \end{equation*}
However, if $J^{N,\lambda}_T(Z)-J^{N,\lambda}_0(Z) < \frac{1}{2}$, we have together with \eqref{Nzero} that $J^{N,\lambda}_T(Z) <1$. And in this case, we can conclude from
\begin{align*} &J^{N,\lambda}_T(Z)-J^{N,\lambda}_0(Z)\\
&= \sup\limits_{0 \leq s \leq T}\left\{ e^{\lambda \sqrt{\text{log}(N)}(T-s)}\left(N^{\delta} \Delta^N_{Z,s} + N^{3\delta -1} \right)\right\} - e^{\lambda \sqrt{\text{log}(N)}T} N^{3\delta -1}\\
&\geq N^{\delta} \sup\limits_{0 \leq s \leq T} \Delta^N_{Z,s} - \frac{1}{2} \end{align*}
the bound
\begin{equation*} \begin{split}\IP_0\Bigl[\sup\limits_{0 \leq s \leq T}\left\{  \Delta^N_{Z,s} \right\} \geq  N^{-\delta}   \Bigr] &\leq \IP_0 \Bigl[J^{N,\lambda}_T(Z)-J^{N,\lambda}_0(Z)
 \geq \frac{1}{2} \Bigr] \\&\leq 2T e^{\lambda \sqrt{\log(N)}T} \frac{C_\kappa}{N^{\kappa}} \leq \frac{2T C_\kappa}{N^{\kappa - 1 + 3\delta}}.
\end{split}\end{equation*}
 For any given $\beta > 0$, we can choose $\kappa := \beta + 1 - 3\delta$ so that the result takes the form
\begin{equation}\IP_0\Bigl[\sup\limits_{0 \leq s \leq T}\left\{  \Delta^N_{Z,s} \right\}
\geq N^{-\delta}   \Bigr] \leq \frac{T C}{N^{\beta}},
\end{equation}
with $C = 2C_\kappa$, as announced. 
\end{proof}

\section{Controlling the mean field dynamics}\label{section_proof2}

\noindent The previous theorem contains our main approximation result for the mean field dynamics. As explained in Section \ref{Section:measurechaos}, two more steps remain in order to complete the proof of Theorem \ref{Thm:Thm}. First, we have to show that the solutions $f^N_t$ of the regularized Vlasov-Poisson equation \eqref{Vlasov} converge to a solution of the proper Vlasov-Poisson equation as the cut-off is lifted with $N \to \infty$. Second we have to prove the approximation of the continuous Vlasov-density by the discretized version $\mu^N_0[\Phi_{t,0}(Z)]$. The proof of $f^N_t \weak f_t$ follows the method of Loeper \cite{Loeper}.

\begin{Proposition}\label{Prop:fNtof}
	Let $f_0$ satisfy the assumptions of Theorem \ref{Thm:Thm}. Let $f^N_t$ and $f_t$ be the solution of the regularized, respectively the unregularized Vlasov-Poisson equation with initial datum $f_0$. Then we have for $p \in [1,\infty]$ and $N > 3$:
	\begin{equation} W_p(f^N_t, f_t) \leq N^{-\delta} \,e^{t2C_0 \sqrt{\log N}},\end{equation}
	where $C_0$ depends on $\sup_{t,N} \lbrace \lVert \rho^N_t \rVert_\infty, \lVert \rho^f_t \rVert_\infty \rbrace$.
\end{Proposition}
\begin{proof} 
	Let $\rho^N_t:=\rho[f^N_t]$ and $\rho^f_t:=\rho[f_t]$ denote the charge density induced by $f^N_t$ and $f_t$, respectively. 
	Let $\varphi^N_t = (Q^N_t, P^N_t)$ be the characteristic flow of $f_t^N$. For the (unregularized) Vlasov-Poisson equation, the corresponding vector-field is not Lipschitz. However, as we assume the existence of a solution $f_t$ with bounded density $\rho_t$, the mean field force $k*\rho_t$ does satisfy a Log-Lip bound of the form $\lvert k*\rho_t(x) - k*\rho_t(y)\rvert \leq C \lvert x-y \lvert (1+\log^-(\rvert x-y\rvert ))$, where $\log^-(x) = \max \lbrace 0, -\log(x) \rbrace$. This is sufficient to ensure the existence of a characteristic flow $\psi^f_{t,s}=(Q^f_{t,s}, P^f_{t,s})$ such that $f_t = \psi^f_{t,s}\# f_s$.\\

	\noindent Since $W_p \leq W_q$ for $p \leq q$, it suffices to prove the statement for the infinite Wasserstein distance $W_\infty$. We consider $\pi_0(x,y) := f_0(x)\delta(x-y) \in \Pi(f_0,f_0)$, which is the optimal coupling yielding $W_\infty(f_t^N, f_t)\lvert_{t=0}=W_\infty(f_0,f_0) = 0$.  Let $\pi_t = (\varphi^N_{t,0}, \psi^f_{t,0})\# \pi_0$. Then $\pi_t \in \Pi(f^N_t, f_t) \, \forall t \in [0,T)$ and we can consider 
	\begin{equation*} \begin{split} 
	D(t) :&= \pi_t\mathrm{-esssup}\, \bigl\lbrace \sqrt{\log(N)} \lvert x^1 - y^1 \rvert+ \lvert x^2 - y^2 \rvert \bigl\rbrace\\ &= \pi_0\mathrm{-esssup} \,\bigl\lbrace \sqrt{\log(N)}\lvert Q^N_{t,0}(x) - Q^f_{t,0}(y) \rvert + \lvert P^N_{t,0}(x) - P^f_{t,0}(y) \rvert\bigl\rbrace\\
	&= f_0\mathrm{-esssup} \,\bigl\lbrace \sqrt{\log(N)}\lvert Q^N_{t,0}(x) - Q^f_{t,0}(x) \rvert + \lvert P^N_{t,0}(x) - P^f_{t,0}(x) \rvert\bigl\rbrace, 
	\end{split}\end{equation*} 
	which is an upper bound on $W_\infty(f^N_t, f_t)$ for $N > 3$. We compute: 
	\begin{equation}\begin{split}\label{dtD}
	\partial_t^+ D(t)  \leq f_0\mathrm{-esssup} \Bigl\lbrace  &\sqrt{\log N}  \lvert P^N_t(x)-P^f_t(x) \rvert \\ &+   \bigl\lvert  k^N_\delta*\rho^N_t (Q_t^N(x)) -  k* \rho^f_t(Q^f_t(x)) \bigr\rvert \Bigr\rbrace\\
	\;\; \leq  \sqrt{\log(N)} D(t)\\+  f_0\mathrm{-esssup} \;  &\bigl\lvert  k^N_\delta*\rho^N_t (Q_t^N(x)) -  k* \rho^f_t(Q^f_t(x))  \bigr\rvert.
	\end{split}\end{equation}
	
	\noindent Now note that 
	\begin{align*} k* \rho^f_t(Q^f_t(x))  &= \int_{\IR^3} k(Q^f_t(x) - q) \, \dd\rho^f_t(q) = \int_{\IR^3 \times \IR^3} k(Q^f_t(x) - q)\,  \dd f_t(q,p) \\
	&= \int_{\IR^3 \times \IR^3} k(Q^f_t(x) - Q^f_t(z)) \dd f_0(z)\end{align*}
	and analogously for $k^N_\delta*\rho^N_t (Q_t^N(x)) $. Hence, we can write the interaction term as	
	\begin{align}\notag 
	&\bigl\lvert  k^N_\delta*\rho^N_t \, (Q^N_t(x)) -  k* \rho^f_t \,(Q^f_t(x)) \bigr\rvert\\\notag 
	&=\Bigl \lvert \int_{\IR^3 \times \IR^3}  \bigl[ k^N_\delta (Q^N_t(x) - Q^N_t(z)) - k(Q^f_t(x) - Q^f_t(z))\bigr]\, \dd f_0(z)  \Bigr \rvert\\\label{termB1}
	&\leq \int_{\IR^3 \times \IR^3} \bigl\lvert k^N_\delta (Q^N_t(x) - Q^N_t(z)) - k^N_\delta (Q^f_t(x) - Q^f_t(z))\bigr\rvert\, \dd f_0(z)\\\label{termB2}
	&+ \int_{\IR^3 \times \IR^3} \bigl\lvert  k^N_\delta (Q^f_t(x) - Q^f_t(z)) - k(Q^f_t(y) - Q^f_t(z))  \bigr\rvert\, \dd f_0(z)
	\end{align}  
	
	\noindent The second term \eqref{termB2} can be bounded as 
	\begin{align}\notag
	\int \bigl\lvert  &k^N_\delta (Q^f_t(y) - Q^f_t(z)) - k(Q^f_t(y) - Q^f_t(z))  \bigr\rvert \dd f_0(z)\\\notag
	&= \int \bigl\lvert  k^N_\delta (Q^f_t(y) - q) - k(Q^f_t(y) - q)  \bigr\rvert \rho^f_t(q) \dd q \\\label{bound}
	&\leq \lVert k^N_\delta - k \rVert_1 \lVert \rho^f_t \rVert_\infty \leq 4 \pi N^{-\delta} \lVert \rho^f_t\rVert_\infty,
	\end{align}
	
	\noindent where we used the fact that $k$ and $k^N_\delta$ differ only on a set of radius $N^{-\delta}$, so that $\lVert k^N_\delta - k \rVert_1 \leq \lVert  k \bigl \lvert_{\lvert q \rvert \leq N^{-\delta}} \rVert_1 =  4 \pi N^{-\delta}$. The term \eqref{termB1} can be bounded by first using the mean-value theorem as
	\begin{align*} 
	\bigl\lvert k^N_\delta &(Q^N_t(x) - Q^N_t(z)) - k^N_\delta (Q^f_t(y) - Q^f_t(z))\bigr\rvert\\
	\leq & \, \Bigl(\lvert \nabla k^N_\delta (Q^N_t(x) - Q^N_t(z)) \rvert + \bigl\lvert \nabla k^N_\delta (Q^f_t(x) - Q^f_t(z))\bigr\rvert \Bigr)\\
	\cdot &\Bigl( \bigl\lvert Q^N_t(x) - Q^f_t(x)\bigr\rvert + \bigl\lvert Q^N_t(z) - Q^f_t(z)\bigr\rvert\Bigr).
	\end{align*}
	Taking the integral with respect to $\dd f_0(z)$, we estimate $\bigl\lvert Q^N_t(z) - Q^f_t(z)\bigr\rvert$ by the esssup and $\lVert \nabla k^N_\delta \rVert_1$ with Lemma \ref{lbound}. This yields: 	
	\begin{equation*}\begin{split} 
	\eqref{termB1} \leq &\; C_0 \log(N) \Bigl( \bigl\lvert Q^N_t(x) - Q^f_t(x)\bigr\rvert + f_0\mathrm{-esssup} \, \bigl\lvert Q^N_t(z) - Q^f_t(z)\bigr\rvert \Bigr).
	\end{split}\end{equation*}
	with $C_0 := 2C_l C_\rho$ as before and $C_\rho$ the uniform bound on the charge density from \eqref{Crho}.	Taking also the $f_0\mathrm{-esssup}$ over the $x$-variable, we get the bound
	\begin{equation*}\begin{split}
	&2 C_0 \log(N) \, f_0\mathrm{-esssup} \, \bigl\lvert Q^N_t(z) - Q^f_t(z)\bigr\rvert \\
	&\leq 2 C_0 \sqrt{\log(N)} D(t).
	\end{split}\end{equation*} 
	Putting everything together, we have found
	\begin{align*}
	\partial_t^+ D(t)\leq  2 C_0 \sqrt{\log(N)} D(t) + 4 \pi  \lVert \rho_t\rVert_\infty N^{-\delta}.
	\end{align*} 
	Using Gronwall's lemma and the fact that $D(0)=0$, we get	
	\begin{equation*} W_\infty(f^N_t, f_t) \leq D(t) \leq e^{t2C_0 \sqrt{\log(N)}} N^{-\delta},
	\end{equation*}
from which the desired proposition follows.
\end{proof}

\noindent A standard Gronwall argument yields the following result.
\begin{Proposition}\label{Prop:pullback}
	Let $\varphi^N_{t} = (Q(t,\cdot), P(t,\cdot))$ the characteristic flow of $f^N_t$ defined by \eqref{chareq2} and ${}^N\Phi_{t,s}$ the lift to the $N$-particle phase-space defined in \eqref{meanfieldflow}. Then we have for all $p \in [1, \infty)$ and $N > 3$:
	\begin{equation}\label{WpResult} W_p(\mu^N_0[\Phi_{t,0}(Z)], f^N_t )
	\leq  \sqrt{\log{N}} \, W_p(\mu^N_0[Z],f_0)\, e^{tC_0\sqrt{\log{N}}}. \end{equation}
\end{Proposition}

\begin{proof}
	For $Z \in \IR^{6N}$ let $\pi_0(x,y) \in \Pi(\mu^N_0,f_0)$ and define  $\pi_t = (\varphi^N_t, \varphi^N_t)\# \pi_0 \in \Pi(\mu^N_0[\Phi_{t,0}(Z)], f^N_t )$. Note that both measures are now transported with the same flow. Set
	\begin{equation*}\begin{split} 
	&D_p(t):= \Bigl[ \int\limits_{\IR^{6} \times \IR^{6}} \Bigl( \sqrt{\log N}\, \lvert x^1-y^1 \rvert + \, \lvert x^2-y^2 \rvert \Bigr)^p\; \dd \pi_t(x,y) \Bigr]^{1/p} \\ 
	=&\Bigl[  \int\limits_{\IR^{6} \times \IR^{6}} \Bigl( \sqrt{\log N}\, \lvert Q_t(x)-Q_t(y) \rvert + \, \lvert P_t(x)-P_t(y) \rvert  \Bigr)^p\,\dd \pi_0(x,y) \Bigr]^{1/p}.
	\end{split}
	\end{equation*}
	We compute:
	\begin{equation*}\begin{split}\label{dtD2}
	\frac{\dd}{\dd t} D^p_p(t) =
	p \int \dd \pi_0(x,y) \, \Bigl(\sqrt{\log N }\, \lvert Q_t(x)-Q_t(y) \rvert + \lvert P_t(x)-P_t(y) \rvert \Bigr)^{p-1}&\\
	\Bigl(\sqrt{\log N}\,  \, \lvert P_t(x)-P_t(y) \rvert +  \, \bigl\lvert  k^N_\delta*\rho^N_t (Q_t(x)) -  k^N_\delta* \rho^N_t(Q_t(y)) \bigr\rvert\Bigr)&
	\end{split}
	\end{equation*}
	
	\noindent Using again the Lipschitz bound 	
	\begin{equation*} 
	\bigl\lvert  k^N_\delta*\rho^N_t (Q_t(x)) -  k^N_\delta* \rho^N_t(Q_t(y)) \bigr\rvert \leq C_0 \log(N) \bigl\lvert Q_t(x) -  Q_t(y) \bigr\rvert 
	\end{equation*}
	for $N > 3$, i.e. $\log(N) > 1$, we arrive at the estimate 	
	\begin{equation}D(t) \leq D(0) +C_0 \sqrt{\log{N}} \int D(s)\, \dd s
	\end{equation} Hence by Gronwall's inequality: 
	\begin{equation*}
	W_p(\mu^N_0[\Phi_{t,0}(Z)], f^N_t )= W_p (\varphi^N_t\#\mu^N_0,\varphi^N_t\# f_t ) \leq  D(t) \leq D(0) e^{tC_0\sqrt{\log{N}}}.
	\end{equation*}
	Taking on the right-hand side the infimum over all $\pi_0(x,y) \in \Pi(\mu^N_0,f_0)$,
	\begin{equation*}
	W_p(\mu^N_0[\Phi_{t,0}(Z)], f^N_t ) \leq \sqrt{\log{N}} \, W_p(\mu^N_0,f_0)\,  e^{tC_0\sqrt{\log{N}}}.
	\end{equation*}
\end{proof}

\noindent In view of \eqref{WpResult}, it remains to establish an upper bound on the typical rate of convergence for $W_p(\mu^N_0[Z], f_0) \to 0$. (Note that, other than that, the result of Proposition \ref{Prop:pullback} is actually deterministic.) Fortunately, we can rely for this purpose on recent, partcularly strong concentration estimates obtained by Fournier and Guillin, 2014 \cite{Fournier}.

\begin{Theorem}[Fournier and Guillin]\label{Fournier}
Let $f$ be a probability measure on $\IR^n$ such that $\exists k>2p$:
	\begin{equation*} 
	M_k(f):=	\int_{\IR^n} \lvert z \rvert^k   \dd f(z) < + \infty.
	\end{equation*}
	Let $(Z_i)_{i=1,...,N}$ be a sample of independent variables, distributed according to the law $f$ and consider $\mu^N_0[Z]:= \sum\limits_{i=1}^N \delta_{Z_i}$. Then, for any $\epsilon > 0$ there exist constants $c,C$ depending only on $k, M_k(f)$ and $\epsilon$ such that for all $N \geq 1$ and $\xi >0$:
	\begin{equation*} 
	\IP_0\Bigl[ W^p_p(\mu^N, f) > \xi \Bigr] \leq C N (N\xi)^{-\frac{k-\epsilon}{p}} +  C\mathds{1}_{\xi \leq 1}\, a(N,\xi),\end{equation*}
	with \begin{equation}	\label{aNxi}
	a(N,\xi):=	
	\begin{cases}\exp(-cN\xi^2) & \text{if } p > n/2\\
	\exp(-cN(\frac{\xi}{\ln(2+1/\xi)})^2) &  \text{if } p = n/2\\
	\exp(-cN\xi^{n/p}) &\text{if } p \in [1,n/2). \end{cases} 
	\end{equation}
\end{Theorem}

\noindent With these large deviation estimates, we get the following.
 
\begin{Corollary}\label{Prop:muphifn}
	Let $p \in [1,\infty)$, $\gamma < \min \lbrace \delta, \frac{1}{6}, \frac{1}{2p} \rbrace$ and $N > 3$. Then there exists constants $c,C >0$ such that
	\begin{equation}\begin{split}\IP_0\Bigl[\exists t\in [0,T] : W_p( \mu^N_0[\Phi_{t,0}(Z)], f_t)  &> (1+\sqrt{\log(N)}) N^{-\gamma} e^{t 2C_0\sqrt{\log N}}\Bigr]\\ &\leq C \bigl( e^{-cN^{1-(6 \vee 2p) \gamma}}  +  N^{1 - \frac{k}{2p}} \bigr)\end{split}\end{equation}
	where we use the notation $6 \vee 2p := \max \lbrace 6 ,2p \rbrace$.  
\end{Corollary} 
\begin{proof} 
	By assumption in Theorem \ref{Thm:Thm}, there exists $k>2p$ such that $M_k(f_0) < +\infty$. Applying Theorem \ref{Fournier} with $\xi = N^{-p \gamma}, \epsilon = \frac{k(1 - 2p \gamma)}{2(1 - p \gamma)}$ and the finite-moment condition (1), we get constants $C,c > 0$ such that
	\begin{equation*}\IP_0\Bigl[ W_p(\mu^N_0[Z], f_0) > N^{- \gamma}\Bigr] \leq C \bigl( e^{-cN^{1-(6\vee 2p)\gamma}}  + N^{1 - \frac{k}{2p}} \bigr).\end{equation*}
	
	\noindent Thus, with Proposition \ref{Prop:pullback}, we conclude 
	\begin{align*}
	\IP_0\Bigl[\exists t\in [0,T] : W_p( \mu^N_0[\Phi_{t,0}(Z)], f^N_t)  > \sqrt{\log(N)} N^{-\gamma} e^{tC_0\sqrt{\log N}}\Bigr]\\
	\leq C \bigl( e^{-cN^{1-(6\vee 2p)\gamma}}  + N^{1 - \frac{k}{2p}} \bigr).
	\end{align*}
Adding the bound $W_p(f^N_t, f_t) \leq N^{-\delta} e^{t 2C_0\sqrt{\log N}}$ from Proposition \ref{Prop:fNtof}, the statement follows.
\end{proof}

\noindent Now we have everything in place to complete the proof of Theorem~\ref{Thm:Thm}.

\begin{proof}[\textbf{\em{Proof of Theorem \ref{Thm:Thm}}}.]
Let $p \in [1, \infty)$, $\gamma < \frac{1}{6}$ and $N >3$. We split the approximation into
\begin{align*}W_p(\mu^N_t[Z] , f_t ) &\leq W_p(\mu^N_0[\Psi_{t,0}(Z)], \mu^N_0[\Phi_{t,0}(Z)])\\
&+ W_p(\mu^N_0[\Phi_{t,0}(Z)], f_t).\end{align*}

\noindent From Corollary \ref{Prop:muphifn} we get constants $c,C_1 >0$ such that
\begin{equation*}\label{loglog}\begin{split}\IP_0\Bigl[\exists t\in [0,T] : W_p( \mu^N_0[\Phi_{t,0}(Z)], f_t) &> (1+\sqrt{\log N}) N^{-\gamma}\, e^{t2C_0\sqrt{\log(N)}} \Bigr] \\
&\leq C_1\bigl(e^{-cN^{1-(6 \vee 2p) \gamma}} +  N^{1 - \frac{k}{2p}} \bigr).\end{split}\end{equation*}


\noindent In Theorem \ref{Thm:Thm1}, we choose $\beta = \frac{k}{2p} - 1$ and get a constant $C'$ so that, together with Lemma \ref{Lemma:maxWinfty},
\begin{equation}
\IP_0\Bigl[\exists t \in [0,T] : W_p(\mu^N_0[\Psi_{t,0}(Z)], \mu^N_0[\Phi_{t,0}(Z)]) \geq N^{-\delta} \Bigr]
\leq  T C' N^{1 - \frac{k}{2p}}
\end{equation}
for any $N \geq  e^{(\frac{C_0 T +1}{1 - 3\delta})^2}$. Putting both estimates together and choosing $\gamma < \min \lbrace \frac{1}{6}, \delta \rbrace$, we have found that
\begin{equation}\begin{split}
\IP_0\Bigl[&\exists t \in [0,T] : W_p(\mu^N_t[Z], f_t) >\\ & N^{-\delta}  +(1+ \sqrt{\log N}) N^{-\gamma}\, e^{t2C_0\sqrt{\log(N)}} \Bigr]\\
\leq & C_1 e^{-cN^{1-(6 \vee 2p) \gamma}}  + C_2 T N^{1 - \frac{k}{2p}},
\end{split}\end{equation}
with $C_2 := C_1 + C'$. We can simplify this result by noting that $e^{\lambda\sqrt{\log(N)}} \leq N^{1 - 3\delta}$ for $N \geq e^{(\frac{\lambda}{1 - 3\delta})^2}$. We shall thus demand $N \geq  N_1:= e^{(\frac{2(C_0 T + 1)}{1 - 3\delta})^2}$ which yields $N^{1 - 3\delta} \geq \max \lbrace  e, 2(1+{\sqrt{\log N})\, e^{2C_0T\sqrt{\log(N)}}} \rbrace$ and conclude that
\begin{equation}\begin{split}
\IP_0\Bigl[\exists t \in [0,T] : W_p(\mu^N_t[Z], f_t) >  N^{-\gamma + 1 - 3\delta}\Bigr]\\
\leq  C_1 e^{-cN^{1-(6 \vee 2p) \gamma}}  + C_2 T \, N^{1 - \frac{k}{2p}}.
\end{split}\end{equation}

\end{proof}

\section{Weaker singularities, open questions}
While the present paper focuses on the Vlasov-Poisson equation, the method presented here can also be applied to interactions with milder singularities, see \cite{Peter}. For better comparison with other approaches, in particular the reference paper of Hauray and Jabin, 2013,\cite{HaurayJabin}, we shall state here the corresponding results without further proof. Generalization to higher dimensions would be straight-forward, as well.

\noindent We use the characterization of force kernels introduced in Def. \ref{Def:conditions}.

\begin{Theorem}
	Let $\alpha < 2$. Let $k$ satisfy a $S^\alpha$ condition and $k^N_\delta$ satisfy a $S^\alpha_\delta$ condition with the additional assumption \eqref{additionalassumption} and 
	\begin{equation}\label{ourcutoff} \delta <  \frac{1}{1+\alpha}.
	\end{equation}
	Assume (for simplicity) that  $f_0 \in  L^1 \cap L^\infty(\IR^3 \times \IR^3, \IR^+)$, normalized to $\int f_0 = 1$, has compact support and let $f$ be the unique solution of the Vlasov equation with force kernel $k$. For $Z \in \IR^{6N}$, let $\mu^N_t[Z]$ the unique weak solution of the (regularized) Vlasov equation with force $k^N_\delta$ and initial data $\mu^N_0[Z]$. Then we have  molecular chaos in the following sense: For any $\beta > 0, \gamma \leq \min \lbrace \frac{1}{6} , \delta \rbrace$ and $T > 0$ there exists constant $C_1, C_2$ such that
	\begin{equation}\label{molchaosweak} \IP_0\bigl[\exists t \in [0,T]:  W_1(\mu_t^N[Z] , f_t) >N^{-\gamma} \bigr] \leq C_1 e^{-cN^{1-6\gamma}}  + T C_2 N^{- \beta}.
	\end{equation}
\end{Theorem}

\noindent This can be compared to the results in \cite{HaurayJabin}, where a statement similar to \eqref{molchaosweak} is derived for the case $1\leq \alpha < 2$ with a cut-off of order
\begin{equation}\label{HJcutoff}\delta < \frac{1}{6}\min \Bigl\lbrace \frac{1}{\alpha-1}, \, \frac{5}{\alpha} \Bigr \rbrace. \end{equation}
For $\alpha \in [1,2)$, the upper bound on $\delta$ given by \eqref{HJcutoff} ranges between $\frac{5}{6}$ and  $\frac{1}{6}$, while our upper bound from \eqref{ourcutoff} ranges between $\frac{1}{2}$ and $\frac{1}{3}$. In particular, it is interesting to note that the cut-off required in \cite{HaurayJabin} is smaller than ours for $\alpha < \frac{7}{5}$ but larger for $\frac{7}{5} < \alpha < 2$. This suggests that the probabilistic estimates presented here fare better for strong singularities  -- in the sense of admitting a significantly smaller cut-off -- while the method proposed in \cite{HaurayJabin} provides better controls for mild singularities.

Most notably, Hauray and Jabin are able to treat the case $ 0 < \alpha < 1$ with no cut-off at all by providing an explicit control on the minimal particle distance (in $(p,q)$-space, strictly speaking, while integrating the forces over small time-intervals). As it stands, our method requires in any case a regularization of the microscopic dynamics. Since it proves very effective in this setting, it would be interesting to investigate if it can be extended  -- or possibly combined with the approach of \cite{HaurayJabin} -- to further reduce the cut-off or, ideally, dispense with it altogether for sufficiently mild singularities.\\

\noindent \textbf{Acknowledgements}\\

\noindent We thank Maxime Hauray, Michael Kiessling, Detlef D\"urr, Martin Kolb, Ana Ca\~nizares, Samir Salem and Young-Pil Choi for valuable comments on earlier versions of the manuscript.

\newpage
\bibliography{VPlit}
\bibliographystyle{plain}
\end{document}